\documentclass{lmcs} %%% last changed 2014-08-20
\pdfoutput=1

\usepackage{lastpage}
\lmcsdoi{21}{1}{25}
\lmcsheading{}{\pageref{LastPage}}{}{}%
{Apr.~05,~2024}{Mar.~17,~2025}{}

\usepackage[utf8]{inputenc} %FIXED Inputencoding
\usepackage{amsthm}
\usepackage{amsmath}
\usepackage{amssymb}
\usepackage{bm}

\usepackage{hyperref}
\theoremstyle{plain} %\crefname{satz}{Satz}{S\"atze}
\begin{document}

\title[Alternating Quantifiers in Uniform One-Dimensional Fragments]{Alternating Quantifiers \texorpdfstring{\\}{} in Uniform One-Dimensional Fragments \texorpdfstring{\\}{} with an Excursion into  Three-Variable Logic}
\titlecomment{{\lsuper*}This paper is a consolidated, revised and extended version of \cite{FK23} and \cite{Kie23}}
%\thanks{thanks, optional.}	%optional

% affiliations are numbered automatically with a, b, c (see below)
% use the optional argument to indicate the affiliation(s) of each author
% omit the argument if there is only one author, or only one affiliation
\author[O.~Fiuk]{Oskar Fiuk\lmcsorcid{0009-0006-1312-4899}}
\author[E.~Kiero\'nski]{Emanuel Kiero\'nski\lmcsorcid{0000-0002-8538-8221}}

% affiliation 1 (automatically numbered a)
\address{Institute of Computer Science, University of Wroc\l{}aw, Poland}	%optional
% write emails for all authors having that affiliation
\email{307023@uwr.edu.pl, emanuel.kieronski@cs.uni.wroc.pl}  %optional

%% etc.

%% required for running head on odd and even pages, use suitable
%% abbreviations in case of long titles and many authors:

%%%%%%%%%%%%%%%%%%%%%%%%%%%%%%%%%%%%%%%%%%%%%%%%%%%%%%%%%%%%%%%%%%%%%%%%%%%

%% the abstract has to PRECEDE the command \maketitle:
%% be sure not to issue the \maketitle command twice!

%%%OUR MACROS

%Logics
\newcommand{\FO}{\mbox{\rm FO}}
\newcommand{\FOt}{\mbox{$\mbox{\rm FO}^2$}}
\newcommand{\FOthree}{\mbox{$\mbox{\rm FO}^3$}}
\newcommand{\TGF}{\mbox{$\mbox{\rm TGF}$}}
\newcommand{\GFcp}{\mbox{$\mbox{\rm GF}^{\times_2}$}}
\newcommand{\GFU}{\mbox{$\mbox{\rm GFU}$}}
\newcommand{\UNFO}{\mbox{\rm UNFO}}
\newcommand{\LGF}{\mbox{\rm LGF}}
\newcommand{\GNFO}{\mbox{\rm GNFO}}
\newcommand{\GFTG}{\mbox{\rm GF+TG}}
\newcommand{\GFEG}{\mbox{\rm GF+EG}}
\newcommand{\TGFTG}{\mbox{\rm TGF+TG}}
\newcommand{\GFUTG}{\mbox{\rm GFU+TG}}
\newcommand{\GF}{\mbox{\rm GF}}
\newcommand{\GFt}{\mbox{$\mbox{\rm GF}^2$}}
\newcommand{\FF}{\mbox{\rm FF}}
\newcommand{\UF}{\mbox{$\mbox{\rm UF}_1$}}
\newcommand{\AUFm}{\mbox{$\mbox{\rm AUF}_1^{-}$}}
\newcommand{\AUF}{\mbox{$\mbox{\rm AUF}_1$}}
\newcommand{\AUFthree}{\mbox{$\mbox{\rm AUF}_1^3$}}
\newcommand{\sUF}{\mbox{$\mbox{\rm sUF}_1$}}
\newcommand{\bK}{\mbox{$\bar{\mbox{\rm K}}$}}
\newcommand{\sUFthree}{\mbox{$\mbox{\rm sUF}_1^3$}}
\newcommand{\FOthreem}{\mbox{$\mbox{\rm FO}^3_{-}$}}

% Complexity classes
\newcommand{\NLogSpace}{\textsc{NLogSpace}}
\newcommand{\LogSpace}{\textsc{LogSpace}}
\newcommand{\NP}{\textsc{NPTime}}
\newcommand{\PTime}{\textsc{PTime}}
\newcommand{\PSpace}{\textsc{PSpace}}
\newcommand{\ExpTime}{\textsc{ExpTime}}
\newcommand{\ExpSpace}{\textsc{ExpSpace}}
\newcommand{\NExpTime}{\textsc{NExpTime}}
\newcommand{\TwoExpTime}{2\textsc{-ExpTime}}
\newcommand{\TwoNExpTime}{2\textsc{-NExpTime}}
\newcommand{\APSpace}{\textsc{APSpace}}
\newcommand{\AExpSpace}{\textsc{AExpSpace}}
\newcommand{\ASpace}{\textsc{ASpace}}
\newcommand{\DTime}{\textsc{DTime}}

% Other 
\newcommand{\set}[1]{\{#1\}}
\newcommand{\str}[1]{{\mathfrak{#1}}}
\newcommand{\restr}{\!\!\restriction\!\!}
\newcommand{\N}{{\mathbb N}}
\newcommand{\Z}{{\mathbb Z}} 
\newcommand{\sss}{\scriptscriptstyle}
\newcommand{\cT}{\mathcal{T}}
\newcommand{\cL}{\mathcal{L}}
\newcommand{\cF}{\mathcal{F}}
\newcommand{\tpstar}{{\rm tp}^*}
\newcommand{\tp}{{\rm tp}}

% AM
\newcommand*{\UU}{\mathsf{U}}
\newcommand*{\NN}{\mathbb{N}}
\newcommand*{\ZZ}{\mathbb{Z}}
\newcommand{\TG}{\mbox{\rm TG}}
\newcommand*{\Lc}{\mathcal{L}}
\newcommand*{\Pc}{\mathcal{P}}
\newcommand*{\Fc}{\mathcal{F}}
\newcommand*{\Mf}{\mathfrak{M}}
\newcommand*{\Nf}{\mathfrak{N}}
\newcommand*{\qf}{\mathrm{qf}}
\newcommand*{\<}{\langle}
\renewcommand*{\>}{\rangle}
\renewcommand{\le}{\leqslant}
\renewcommand{\ge}{\geqslant}
\renewcommand{\bar}{\overline}

\newcommand{\noproof}{\pushQED{\qed} \qedhere \popQED}

\newcommand{\Qfr}{\mbox{Q}}

\newcommand{\phie}{\phi^{\sss \exists}}
\newcommand{\phiu}{\phi^{\sss \forall}}

\newcommand{\psie}{\psi^{\sss \exists}}
\newcommand{\psiu}{\psi^{\sss \forall}}
\newcommand{\mse}{m_{\sss \exists}}
\newcommand{\msu}{m_{\sss \forall}}

\renewcommand{\phi}{\varphi}

\newcommand{\type}[2]{{\rm tp}^{{#1}}({#2})}

\newcommand{\AAA}{\mbox{\large \boldmath $\alpha$}}
\newcommand{\BBB}{\mbox{\large \boldmath $\beta$}}
\newcommand{\CCC}{\mbox{\large \boldmath $\gamma$}}

\newcommand{\pretype}[2]{{\rm pretp}^{{#1}}({#2})}
\newcommand{\outertype}[2]{{\rm otp}^{{#1}}({#2})}

\newcommand{\ff}{\mathfrak{f}}
\newcommand{\fg}{\mathfrak{g}}
\newcommand{\fh}{\mathfrak{h}}
\newcommand{\pat}{\mathfrak{pat}}
\newcommand{\chc}{\mathfrak{chc}}
\newcommand{\eki}[1]{{\color{blue} \tiny \bf Emanuel says: #1}}
\newcommand{\ofi}[1]{{\color{green} \tiny \bf Oskar says: #1}}
\renewcommand{\mod}{\text{mod }}
\newcommand{\gtp}{{\rm gtp}}
\newcommand{\Vars}{{\rm Vars}}
\newcommand{\Prob}{{\rm Prob}}

\begin{abstract}

The uniform one-dimensional fragment of first-order logic was introduced a few years ago as a generalization
of the two-variable fragment to contexts involving relations of arity greater than two. 
Quantifiers in this logic are used in blocks, each block consisting only of existential quantifiers or only of universal 
quantifiers. In this paper we consider the possibility of mixing both types of quantifiers in blocks. We show the finite
(exponential) model property and \NExpTime-completeness of the satisfiability problem for two restrictions of the resulting
formalism: in the first we require that every block of quantifiers is either purely universal or ends with the existential
quantifier, in the second we restrict the number of variables to three; in both equality is not allowed.
We also extend the second variation to a rich subfragment of the three-variable fragment (without equality) that still has
 the finite model property and decidable, \NExpTime-complete, satisfiability.

\end{abstract}

\maketitle

\section{Introduction}

In this paper we aim to advance the research on the uniform one-dimensional fragment of first-order logic, \UF{}.
Additionally, we will take an excursion to the area of three-variable logic, and get close to the limits of decidability.

As \UF{} was originally proposed as a generalization of the two-variable fragment, \FOt{}, to contexts with relations of arbitrary arity,
we set up the scene by recalling the most important facts about \FOt{}.

\smallskip
\noindent
{\bf The two-variable fragment.}
\FOt{} is obtained just by restricting first-order logic so that formulas may use only variables $x$ and $y$,
and is considered as one of the most significant decidable fragments of first-order logic identified so far.
The decidability of its satisfiability problem was shown by Scott \cite{Sco62} in the case without equality, and
by Mortimer \cite{Mor75} in the case with equality. Actually, in \cite{Mor75} it is proved that \FOt{} has the finite model property, that is,
every satisfiable \FOt{}-sentence has a finite model. Later, Gr\"adel, Kolaitis and Vardi \cite{GKV97} strengthened that result by showing the exponential model property of \FOt{}, that is, every satisfiable \FOt{}-sentence has a model of size bounded exponentially in its length. The latter result establishes that the satisfiability problem for \FOt{} is in \NExpTime{}. The matching lower bound follows from the earlier work by Lewis \cite{Lew80}.
With respect to the number of variables \FOt{} is the maximal decidable fragment of \FO{}, as already the three-variable fragment, \FOthree{},
has undecidable satisfiability \cite{KMW62}.

An important  motivation for studying \FOt{} is the fact that it embeds, via the so-called standard translation, basic modal logic and many standard description logics. 
Thus, \FOt{} constitutes an elegant first-order framework for analysing those formalisms.
Its simplicity and naturalness make it also an attractive logic in itself, with applications in knowledge representation, artificial intelligence, and verification of hardware and of software.
In the last few decades, \FOt{} together with its extensions and variations have been extensively studied, and plenty of results have been obtained. To mention just a few of them: decidability was shown for \FOt{} with counting quantifiers \cite{GOR97,PST97,P-H10}, one or two equivalence relations \cite{KO12,KMPHT14}, counting quantifiers and equivalence relation \cite{PH15}, betweenness relations \cite{KLPS20}, and its complexity was established on words and trees, in various scenarios, including the presence of data or counting \cite{BMS09,BDM11,CW16,CW16a,BBC16}.

\smallskip\noindent
{\bf Decidable fragments embedding FO$^{\bm{2}}$.}
Further applications of \FOt{}, e.g., in database theory, are limited, as \FOt{} and its extensions mentioned above can speak non-trivially only about
relations of arity at most two. This restriction is not present in many other decidable fragments studied because of their potential applications in computer science, e.g., the guarded fragment, \GF{} \cite{ABN98},
the unary negation fragment, \UNFO{} \cite{StC13}, the guarded negation fragment, \GNFO{} \cite{BtCS15}, or the fluted fragment, \FF{} \cite{Q69,P-HST19}. 

None of the above-mentioned fragments contains \FOt{}. Hence, a natural question arises: Is there an elegant decidable formalism which retains full expressivity of \FOt{}, but additionally, allows one to speak
non-trivially about relations of arity bigger than two? In the recent literature, several proposals have emerged in this direction. 

One of them is an idea of combining \FOt{} with \GF{}, which can be found in Kazakov's PhD thesis \cite{Kaz06}, and was also present in the work
by Bourish, Morak and Pieris \cite{BMP17}. In the end, it received a more systematic treatment in the paper by 
Rudolph and \v{S}imkus \cite{RS18}, who formally introduced the triguarded fragment, \TGF{}, which is obtained from
\GF{} by allowing quantification for subformulas with at most two free variables to be unguarded.
This yields a logic in which one can speak freely about pairs of elements and in a local, guarded way about tuples of bigger arity. 
\TGF{} turns out to be undecidable with equality, but becomes decidable when equality is forbidden. The satisfiability problem
is then \TwoExpTime{}- or \TwoNExpTime-complete, depending on whether constants are allowed in signatures \cite{RS18}; the finite model property is retained \cite{KR21}.
A variation of the idea above is the one-dimensional triguarded fragment \cite{Kie19}, still containing \FOt{}, which is decidable even in the presence of equality.

\FOt{} (or, actually, even its extension with counting quantifiers, C$^2$) combined with \FF{} was demonstrated to be decidable by Pratt-Hartmann \cite{P-H21}, but the complexity of its satisfiability problem is non-elementary,
as already \FF{} alone has non-elementary complexity
\cite{P-HST19}. Recently, Bednarczyk, Kojelis and Pratt{-}Hartmann proposed a generalization of \FF, called the {adjacent fragment} \cite{BKP23},
which also embeds \FOt{}.

Finally, probably the most canonical generalization of \FOt{} to contexts with relations of arbitrary arity, in our opinion capturing the spirit of \FOt{} more faithfully than
the formalisms discussed above, is the \emph{uniform one-dimensional fragment}, \UF{}, proposed by Hella and Kuusisto \cite{HK14}.

In \UF{} quantifiers are used in blocks, each block consists exclusively of existential or universal quantifiers, and leaves at most one variable free;
a fragment meeting this condition is called \emph{one-dimensional}. However, one-dimensionality alone is not sufficient for the decidability of the satisfiability problem \cite{HK14}, and thus another restriction is needed. In the case of \UF{}, this is the \emph{uniformity} restriction, which, roughly speaking, permits 
to use atoms with at least two variables in the same Boolean combination only if they use precisely the same set of variables (atoms using only a single variable can be used freely).
In effect, just as \FOt{} contains modal logic (or even Boolean modal logic), \UF{} contains \emph{polyadic} modal logic (even with negations of the accessibility relations, cf.~\cite{Kuu16}).
In \cite{HK14} it is shown that \UF{} without equality is decidable and has the finite model property.
In \cite{KK14} this result is strengthened by proving 
that \UF{}, even with a \emph{free use} of equalities, has the exponential model property and thus \NExpTime-complete satisfiability, exactly as \FOt{}.
Here, \emph{free use} means that the uniformity restriction does not constrain the use of equality symbol.

We note that \FOt{} naturally satisfies both restrictions imposed in \UF{}.
The requirement of one-dimensionality is fulfilled whenever a quantifier binds a variable, such as $x$, automatically leaving only $y$.
Similarly, the uniformity restriction is trivially met, since the set $\{x,y\}$ admits only singleton non-empty proper subsets, implying that non-singleton atoms must precisely mention it.

\smallskip\noindent
{\bf Undecidability of FO${^{\mathbf{3}}}$.}
We already mentioned that the three-variable fragment is undecidable. The first undecidability proof is due to Kahr, Moore and Wang \cite{KMW62}.
More precisely, they demonstrated that the class of formulas of the shape $\forall x \exists y \forall z \psi(x,y,z)$, with $\psi$ being quantifier-free and without equality, is undecidable.
Apart from that, there are two other sources of undecidability of \FOthree{} that are worth highlighting here.
The first one is that in \FOthree{} it is possible to specify that some binary relations are transitive (by the obvious formulas starting with the quantification pattern $\forall\forall\forall$). This leads us to undecidability, as already \FOt{} (and even its guarded subfragment) with two transitive relations is undecidable (see Kazakov's Phd thesis \cite{Kaz06} or \cite{Kie05}).
The second one is that the class of formulas admitting equality and of the shape $\forall\forall\exists \psi(x,y,z)$, with $\psi$ being quantifier-free,
constitutes a three-variable subfragment of the G\"odel class with equality, which is known to be undecidable from the work of Goldfarb \cite{Gol84}.
 
\smallskip\noindent
{\bf Our contribution.}
We aimed to discover attractive extensions of \UF{}, enabling the expression of even more interesting properties while preserving the decidability of satisfiability.
As we have previously discussed, both \UF{} and \FOt{} are already quite robust and significant fragments of \FO{}. Therefore, studying their extensions is a well-justified research direction in its own right.
Moreover, this can help us to better understand the limits of decidability, a long-standing theme in computer science.
We also believe that novel formalisms might be useful in practice by offering powerful extensions of modal and description logics to contexts involving relations of arity greater than two.

In the definition of \UF{}, there is a stipulation that each block of quantifiers must exclusively use either universal or existential quantifiers.
In our work we are interested in the following  questions:
Is this requirement necessary for decidability? What implications arise for the complexity of satisfiability if we permit to mix quantifiers inside blocks?

For illustrative purposes, consider the following property \emph{``every director has an actor who played in all his famous movies''}, which can be expressed in \FO{} as follows:
\begin{align}\label{ex1}
  \forall d \exists a \forall m \big[\big(\text{\em Director}(d) \wedge \text{\em FamousMovie}(m)\big) \rightarrow \big(\text{\em Actor}(a) \wedge \text{\em PlayedIn}(a,m,d)\big)\big].
\end{align}
As one may see, despite the fact that this formula satisfies the uniformity condition, it does \emph{not} belong to \UF{}, as the subformula starting with ``$\forall m$'' has two free variables: $d$ and $a$. However, if we were able to mix quantifiers within blocks, allowing us to treat ``$\forall d \exists a \forall m$'' as a single block, then the one-dimensionality restriction would not be violated.

Another natural, although more complicated, property, also not directly expressible in \UF{}, is \emph{``for every Polish team there is a referee who refereed matches of this team against every Italian team''}, which can be translated into \FO{} as follows:
\begin{align}\nonumber
  \forall p & \big[\text{\em PolishTeam}(p) \rightarrow \exists r \forall i \exists a \exists n \big(\text{\em Referee}(r) \wedge \text{\em Arena}(a) \\
	\label{ex2}
  & \wedge \text{\em MatchNumber}(n) \wedge [\text{\em ItalianTeam}(i) \rightarrow \text{\em PlayedMatch}(p,i,r,a,n)] \big) \big].
\end{align}
Once again, in this formula, if quantifiers could be mixed, we would have two blocks: ``$\forall p$'' and ``$\exists r \forall i \exists a \exists n$''. Notice that the subformula beginning with the second block has only one free variable, $p$, and hence obeys the restriction of one-dimensionality.

Motivated by the examples above, we introduce an expressive variation of \UF{} allowing us to mix both types of quantifiers within blocks: the \emph{alternating uniform one-dimensional fragment}, \AUF{}. 

We remark here that  the definition of  \UF{}/\AUF{}, even though in general lead to a different formalisms, show some similarities with the concept of \emph{binding form} \cite{MP15,BM17}. We, however, do not study this issue in more detail, noting here only that logics defined via binding forms have lower computational complexity, fitting into the polynomial hierarchy.

Unfortunately, we could not determinate whether the whole \AUF{} has the finite model property, or even if its satisfiability problem is decidable.
However, we successfully managed to accomplish this task for its two non-trivial subfragments, denoted \AUFm{} and \AUFthree{}, both without equality.
\AUFm{}, called the \emph{diminished alternating uniform one-dimensional fragment}, is obtained from \AUF{} by restricting the shape of quantifier blocks: each block either consists solely of universal quantifiers or it ends with the existential one; this idea is inspired by the definition of the Maslov class $\bar{\text{K}}$ \cite{Mas71}. \AUFthree{} is just the restriction of \AUF{} to its three-variable subfragment. 

Our main technical result is that both these subfragments, \AUFm{} and \AUFthree{}, have the exponential model property, that is, their every satisfiable formula has a model of size bounded exponentially in its length. This immediately yields the decidability and \NExpTime-completeness of their satisfiability problems.
Worth mentioning here is that the result for \AUFm{} can be also obtained via a simple embedding in the class of conjunctions of Maslov class $\bar{\text{K}}$ sentences (see Section~\ref{s:aufmany}),
whose satisfiability was shown to be \NExpTime-complete in our recent work \cite{FKM24}. 
The case of $\bar{\text{K}}$/$\bar{\text{DK}}$ is, however, much more involved than the case of \AUFm{}, and in this paper 
we present a simpler, direct approach for the latter.

After we analyse $\AUFthree{}$, we take the occasion and make a step towards a better understanding of the undecidability
of \FOthree{}. For this purpose, we introduce a rich, decidable subfragment of \FOthree{}, containing \AUFthree, and hence also \FOt{} (without equality).
Roughly speaking, our new fragment, \FOthreem{}, is defined by limiting (but not forbidding!) the use of two ``dangerous'' quantifier
patterns: $\forall\exists\forall$ and $\forall\forall\forall$, identified in the paragraph on the undecidability of \FOthree{}.
As we restrict ourselves to the case without equality, the third pattern $\forall\forall\exists$ does not lead directly to undecidability, and hence we do not consider it as ``dangerous''.

\FOthreem{}, in addition to what is expressible in \AUFthree{}, allows us to specify new properties not being expressible in other discussed fragments.
An example of such a property is ``\emph{every married couple has a child}'', with its natural translation into \FO{}:
\begin{align} \label{ex3}
  \forall h \forall w \big[\text{\em Married}(h,w) \rightarrow \exists c \big( \text{\em ChildOf}(c,h) \wedge \text{\em ChildOf}(c,w)\big)\big].
\end{align}
We stress that the subformula starting with ``$\exists c$'' is neither one-dimensional nor uniform. However, this example is contained in \FOthreem{}.
Hence, indeed, \FOthreem{} is more expressive than \FOt{}, even over signatures built out only of unary and binary relation symbols,
as demonstrated by the above formula (we remark that over such signatures \AUFthree{} syntactically collapses to \FOt{}).
Therefore, we believe that \FOthreem{} can serve as a foundational framework for an interesting specification language, distinct from other proposals, giving a lot of freedom for speaking about relations of arity at most three.

\smallskip
\noindent
{\bf Organization of the paper.} The organization of this paper is as follows.

In Section~\ref{s:preliminaries}, we introduce our notation and terminology, formally define the relevant fragments of first-order logic, and present a normal form for \AUF{}.

In Section~\ref{s:aufmany}, we establish that \AUFm{} possesses the exponential model property.
We provide two proofs of the existence of exponential size models of satisfiable \AUFm{}-sentences:
the first one is constructive and involves a use of Hall's marriage theorem, the second one goes via the probabilistic method.
 
In Section~\ref{s:auf}, we show that \AUFthree{} has the exponential model property.
We first solve its subclass in which we focus mainly on the quantifier prefix $\forall\exists\forall$, and then extend the construction to capture full \AUFthree{}. 

In Section~\ref{s:fothreem}, we study the decidable extension \FOthreem{} of \AUFthree{}.
This section has two independent parts. In Subsection~\ref{ss:complexity}, we establish the complexity of \FOthreem{} by
showing that for a given satisfiable \FOthreem{}-sentence there is an exponentially bounded structure which can be unravelled into an infinite model, and thus may serve as a \emph{witness for satisfiability}.
In Subsection~\ref{ss:fmp}, we combine classical style constructions from Section~\ref{s:auf} with the probabilistic method to obtain the finite, exponential model property for \FOthreem{}.

In Section~\ref{s:conc}, we conclude our work by providing comments on the equality, and by discussing possible directions for further research on uniform one-dimensional fragments.

With such an organization, certain sections of this paper may appear redundant.
Indeed, in Section~\ref{s:aufmany}, we provide two different proofs of the exponential model property for \AUFm{}.
Further, overlap is also present in the part of our paper concerning three-variable logic,
as the complexity results of \FOthreem{} from Subsection~\ref{ss:complexity} are also implied by the probabilistic proof of the exponential model property from Subsection~\ref{ss:fmp}.
Finally, since \FOthreem{} subsumes \AUFthree{}, and both of them have the same \NExpTime{} complexity, results from Section~\ref{s:fothreem} cover those from Section~\ref{s:auf}.

However, in our opinion these redundancies are well-justified.
In the case of \AUFm{}, our proofs utilize two distinct views on the problem.
The classical construction highlights the combinatorial structure of models (solved via Hall's marriage theorem),
while the probabilistic proof reveals the quantitative results (that is almost every structure coming from some ``nice'' family of structures is a model).
For a similar aim of gaining deeper insight, this time into the reasons for the decidability of \FOthreem{},
we decided to retain Subsection~\ref{ss:complexity} with an alternative approach for identifying the complexity.
On the contrary, the overlap between \AUFthree{} and \FOthreem{} is mostly to help the reader to understand our ideas in a simpler setting first (as the proofs of the latter are rather technically involved).
Moreover, this allows us to obtain modularization of the proof for \FOthreem{}, as we use some crucial part from the proof for \AUFthree{} as a black box.

\section{Preliminaries} \label{s:preliminaries}

\subsection{Fragments of logic} \label{s:fragments}

We assume that the reader is familiar with the syntax of first-order logic.
Throughout the remaining, we work with signatures consisting solely of relational symbols, thus we exclude function and constant symbols.
We also restrict ourselves to equality-free first-order logic until the final discussion in Section~\ref{s:conc}.

We briefly recall that the \emph{satisfiability problem} for a fragment of first-order logic asks whether a sentence from it has a model.
We say that a fragment has the \emph{finite model property} if its every satisfiable sentence has a finite model.
Further, we speak about the \emph{exponential model property} if its every satisfiable sentence $\phi$ has a model bounded by $2^{p(|\phi|)}$, with $p(x)$ being a polynomial, and $|\phi|$ being the length of $\phi$, measured in any reasonable fashion. On the contrary, the size of a model is just the size of its domain.

\smallskip
\noindent
{\bf AUF$_{\rm{\mathbf{1}}}$.}
Now, we define the \emph{alternating uniform one-dimensional fragment}, \AUF{}, playing a central role in our investigations.
Informally speaking, in this fragment (i) blocks of quantifiers leave at most one variable free, and 
(ii) atoms with at least two variables can be used in the same Boolean combination only if they use
precisely the same set of variables (atoms with one variable can be used freely). Let us turn to a formal definition.

We define an atom as a $V$-\emph{atom} if its set of variables is exactly $V$.
Then the set of \AUF{}-formulas is defined as the smallest set such that:
\begin{enumerate}[label=(\roman*), align=left] %FIXED changed from itemize to enum with roman numerals to maintain consistant spacing
\item Every atom using at most one variable is in \AUF{}.
\item \AUF{} is closed under Boolean combinations.
\item Let $V$ be a set of variables, and $v_1, \ldots, v_k$ be one of its enumerations.
  If $\psi(v_1,\dots,v_k)$ is a Boolean combination of $V$-atoms and \AUF{}-formulas with free variables in $V$,
  then $\Qfr_1 v_1, \ldots, \Qfr_k v_k \psi(v_1,\dots,v_k)$ and $\Qfr_2 v_2, \ldots, \Qfr_k v_k \psi(v_1,\dots,v_k)$, where the $\Qfr_i$'s are quantifiers, both belong to \AUF{}.  
\end{enumerate}
\noindent %FIXED indent
For a clearer understanding of point (iii) in the above definition, consider the following example (here we focus only on syntactic issues; this formula is not very meaningful):
\begin{align}\nonumber\label{ex4}
  \exists x \forall y \exists z
  & \big[ T(x,y,x,y,z) \vee P(x) \vee S(x,x) \vee \forall x \exists z \big(\neg R(x,y,z)\big) \\
  & \vee \forall x\forall y \big(\neg S(y,x) \wedge P(x) \big) \vee \exists x \exists z \big(R(z,y,x)\big) \big].
\end{align}
In this example: $P(x)$ and $S(x,x)$ are $\{x\}$-atoms, which are by point (i) also \AUF{}-subformulas, and $T(x,y,x,y,z)$ is an $\{x,y,z\}$-atom.
We have also the following \AUF{}-subformulas with quantifiers: $\forall x \exists z \big(\neg R(x,y,z)\big)$ and $\exists x \exists z \big(R(z,y,x)\big)$, which both leave one variable free, $y$; and $\forall x\forall y \big(S(y,x) \wedge P(x) \big)$, which is a closed subformula.
We obtain the whole formula by applying rule (iii), with the quantifier sequence $\Qfr_1,\Qfr_2,\Qfr_3$ being ``$\exists\forall\exists$'', to the disjunction of the mentioned \AUF{}-subformulas and a $\{x,y,z\}$-atom $T(x,y,x,y,z)$.
Notice that the atom $T(x,y,x,y,z)$ cannot be replaced in this example by, say, an atom $T(x,y,x,y,x)$, as it does not use precisely all the variables $\{x,y,z\}$.
The reader can find two other examples of \AUF{}-formulas in Introduction: (\ref{ex1}) and (\ref{ex2}).

After becoming acquainted with the definition of \AUF{}, we proceed to discuss its relevant subfragments.
First, we remark that \AUF{} is defined, mostly for the clarity of presentation, \emph{not} as an extension of \UF{} originally defined in \cite{HK14},
but rather as an extension of the \emph{strongly} uniform one-dimensional fragment, \sUF{}, introduced in \cite{KK15} and being a subfragment of \UF{}.
Formally, \sUF{} is obtained by modifying point (iii) of the definition of \AUF{}: we are obligated to use the quantifiers of the same type for the $\Qfr_i$'s.
Here is an example of a formula in \sUF{} (again, only syntax matters):
\begin{align}\label{ex5}
  \forall x \forall y \forall z \big[P(x) \vee R(y,y,y) \vee R(z,y,x) \vee \exists y \exists z \big(\neg R(y,z,x) \wedge \neg R(x,y,z) \wedge P(y)\big) \big].
\end{align}
In this case, we have two blocks: universal one, ``$\forall x\forall y\forall z$'', and existential one, ``$\exists y\exists z$''.

For interested readers, we explain the original \emph{uniform one-dimensional fragment}, \UF{}: it is defined similarly to \sUF{} --- every block must use the same type of quantifiers, but, in point (iii) of the definition,
the non-unary atoms do \emph{not} have to use the whole set $\{v_1, \ldots, v_k\}$ of variables, but rather all those atoms use the same subset of this set,
e.g., in (\ref{ex5}), we can replace $R(y,z,x)$ and $R(x,y,z)$ with, say, $R(x,z,z)$ and $R(x,x,z)$, so they both use the same subset of variables $\{x,z\}$ in the subformula beginning with ``$\exists y\exists z$'' (see \cite{KK15} for further discussion on variations of \UF{}).
Our approach could perhaps be adapted to capture the full \UF{}, but the technical details would then become awkward.

Another subfragment, which we consider, is \AUFm{}: it consists of \AUF{}-formulas in negation normal form
which can be generated if we restrict the point (iii) of the definition of \AUF{} to allow us to use only \emph{diminished} sequences for the $\Qfr_i$'s.
We say that a sequence of quantifiers is \emph{diminished} if it is either purely universal or ends with the existential quantifier.
Notice that (\ref{ex4}) is a valid \AUFm{}-formula, as its every maximal block of quantifiers is diminished: ``$\exists\forall\exists$'', ``$\forall\exists$'', ``$\exists\exists$'' end with the existential quantifier, and ``$\forall\forall$'' is a purely universal block.
Here is another example of an \AUFm{}-formula:
\begin{align}
\forall x \forall y \forall z \big[ R(x,y,z) \vee \exists t T(x,t) \wedge \exists t T(y,t) \wedge \exists t T(z,t) \big].
\end{align}
On the contrary, (\ref{ex1}) is \emph{not} an \AUFm{}-formula, as it begins with the quantifier block: ``$\forall\exists\forall$''.

The last subfragment of our interest is \AUFthree{}: the three-variable subfragment of \AUF{}.
Its example formula is  (\ref{ex4}), as it uses three-variables: $x,y,z$. Notice that we can reintroduce variables freely.

\smallskip
\noindent
{\bf FO$^{\mathbf{3}}_{\mathbf{-}}$.}
Finally, we leave the realm of \AUF{} and consider a broad subfragment of \FOthree{}:
\FOthreem{}, containing full \FOt{} (without equality), and reaching into the area of \FOthree{} even further than \AUFthree{}.
Its definition may seem very technical. However, our motivations become clear in Section~\ref{s:fothreem}, where we introduce its normal form.
To avoid dealing with some distracting details, we restrict our attention to formulas in negation normal form.
Formally, the set of \FOthreem{}-formulas is the smallest set of formulas over variables $x,y,z$ such that:
\begin{enumerate}[label=(\roman*), align=left]%FIXED same fix as other enumerate
\item Every literal using at most one variable is in \FOthreem{}.
\item \FOthreem{} is closed under conjunction and disjunction.
\item Let $v$ be a variable. If $\psi$ is a positive Boolean combination of \FOthreem{}-formulas and literals, then $\exists v \psi$ belongs to \FOthreem{}.
\item Let $v,v'$ be distinct variables. If $\psi$ is a positive Boolean combination of \FOthreem{}-formulas with
  free variables in $\{v, v'\}$ and literals using the variables $v$ and $v'$, then $\forall v \psi$ is in \FOthreem{}.
\item Let $v,v'$ be any variables.
  If $\psi$ is a positive Boolean combination of \FOthreem{}-formulas with at most one free variable and literals using precisely all of $x,y,z$,
  then $\exists v \forall v' \psi$ and $\forall v \forall v' \psi$ both belong to \FOthreem{}.  
\end{enumerate}
\noindent %FIXED indent
We remark that in the above definition $v$ and $v'$ always represent variables from the set $\{x, y, z\}$.
The main restriction on \FOthree{} formulas, which we impose in (v), is limiting the use of quantifier patterns ending with the universal quantifier and binding subformulas with three free variables.
On the contrary, existential quantification in (iii) and universal quantification for subformulas with two free variables in (iv) can be used quite freely.

It is readily verified that the negation normal form of any sentence in \FOt{}, or in \AUFthree{}, is indeed in \FOthreem{}. 
For example, to show that $\forall x \forall y \neg R(x,y)$ is in $\FOthreem$, we first use rule (iv) (with $v:=y$ and $v':=x$) to generate $\forall y \neg R(x,y)$,
and then again use rule (iv) (with $v:=x$ and $v':=y$). To show that  $\forall x \exists y \forall z P(x,y,z)$ is in \FOthreem{}, we first generate 
$\exists y \forall z P(x,y,z)$ using rule (v) (with $v:=y$, $v':=z$), and then use rule (iv) (with $v:=x$ and arbitrary $v'$).\\ %FIXED added linebreak to remove isolated 'A'
A more complicated example of a formula in \FOthreem{} is left for the reader to contemplate:\vspace*{0.5\baselineskip}
\begin{align}\forall x \forall y \big[\neg R(x,y) \vee \exists z \big(S(x,z) \wedge S(y,z) \wedge [\neg T(x,z,y) \vee \forall x S(x,z) \vee \exists x \forall y (T(z,y,x) {\wedge} P(y))]\big)\big].\notag %\raisetag{-\baselineskip}
\end{align}
\\[0.5\baselineskip]%FIXED Added vertical Space and removed Eq tag to reduce clutter as its not being used. Alternative option is with Tag
Additionally, two simpler examples are present in the Introduction (see (\ref{ex1}) and (\ref{ex3})).

We finish this subsection with a comment on the non-symmetric treatment of universal and existential quantifiers in \AUFm{} and \FOthreem{},
which implies that these formalisms are not closed under negation.
This is in contrast to, e.g., \FOt{}, \UF{} and \AUF{}, but similar to some other known decidable classes of first-order formulas, like the prefix classes or the unary negation fragment, \UNFO{}, \cite{StC13}.

\subsection{Notation, outer-structures and outer-types}
We refer to structures using Fraktur capital letters, and to their domains using the corresponding Roman capitals.
If $\str{A}$ is a structure and $B$ is a subset of $A$, then we denote by $\str{A} \restr B$ the restriction of $\str{A}$ to its subdomain~$B$. 

We usually use $a, b, \ldots$ to denote elements from the domains of structures, and $x$, $y$, $\ldots$ for variables, all of these possibly with some decorations.
For a tuple of elements $\bar{x}$, we write $\psi(\bar{x})$ to denote that all the free variables of $\psi$ are in $\bar{x}$.

In the context of  uniform fragments, it is convenient to speak about some partially defined (sub)structures which we will call {outer-(sub)structures}.
An \emph{outer-structure} over a signature $\sigma$ consists of its domain $A$ and a function specifying the truth-value of every fact
$P(\bar{a})$, for $P \in  \sigma$ and a tuple $\bar{a}$ of elements of $A$ of length equal to the arity of $P$, such that 
$\bar{a}$ either contains all the elements of $A$ or just one of them. The truth values of all the other facts remain unspecified.
Every structure naturally induces an outer-structure. Also, given a structure $\str{A}$ and its subdomain $B$ we denote by
$\str{A} \restr_* B$ the \emph{outer-substructure} of $\str{A}$ with domain $B$, that is the induced outer-structure of $\str{A} \restr B$.

An (atomic) $1$-{\em type} over a signature $\sigma$ is a $\sigma$-structure over the domain $\{1 \}$. A $k$-\emph{outer-type} is an outer-structure over $\{1, \ldots, k\}$.
Let $\str{A}$ be a structure and $a \in A$ an element of its domain. We denote by $\type{\str{A}}{a}$ the $1$-type \emph{realized} in $\str{A}$ by $a$, \emph{i.e.}, the $1$-type  
obtained from $\str{A} \restr \{a\}$ by renaming the element $a$ to $1$.
Let $a_1, \ldots, a_k \in A$ be pairwise distinct elements. We denote by  $\outertype{\str{A}}{a_1, \ldots, a_k}$ the $k$-outer-type \emph{realized} in $\str{A}$ by
the tuple $a_1, \ldots, a_k$, \emph{i.e.}, the $k$-outer-type obtained from $\str{A} \restr_* \{a_1, \ldots, a_k\}$ by renaming the element $a_i$ to $i$, for $i=1, \ldots, k$. 
We remark that the notions of $1$-types and $1$-outer-types coincide.

Note that to fully specify a structure over a signature containing relation symbols of arity at most $\ell$, it suffices to define its domain $A$, and, for every subset $B$ of $A$ of size at most $\ell$, the outer-type of some tuple consisting of all the elements of $B$. 

In our proofs we will need to bound the number of all possible $k$-outer-types over a given signature. For this we we state the following lemma. We skip its obvious proof.
\begin{lem}\label{l:ntyp}
  Let $\sigma$ be a purely relational finite signature containing relational symbols of 
  arity at most $\ell$, for some fixed $\ell$. Then the number of all possible $k$-outer-types, for $k \le \ell$, is an exponential function of $|\sigma|$.
  This holds in particular for the number of $1$-types.
\end{lem}

\smallskip\noindent
\emph{The pull operation.} When defining logical structures, we will heavily use the following \emph{pull operation}. 
Let $\str{C}$ be a partially defined structure, and let $\str{A}$ be a fully defined structure.
If $c_1, \ldots, c_k$ is a sequence of {\em pairwise distinct} elements of $C$, and $a_1, \ldots, a_k$ is a sequence of {\em not necessarily pairwise distinct} elements of $A$,
then we write $\str{C}\restr \{c_1, \ldots, c_k\} \leftarrow pull(\str{A}, c_1 \leftarrow a_1, \ldots, c_k \leftarrow a_k)$ to specify that $\str{C}\restr \{c_1, \ldots, c_k\}$ is copied from $\str{A}\restr \{a_1, \ldots, a_k\}$.
Formally, for every relation symbol $R$, and every sequence of indices $i_1, \ldots, i_\ell$, where $\ell$ is the arity of $R$,
we set $\str{C} \models R(c_{i_1}, \ldots, c_{i_\ell})$ iff $\str{A} \models R(a_{i_1}, \ldots, a_{i_\ell})$.

We also introduce a variant of the pull operation for specifying outer-substructures.
Under the assumptions as stated above, we write $\str{C}\restr_* \{c_1, \ldots, c_k\} \leftarrow pull^*(\str{A}, c_1 \leftarrow a_1, \ldots, c_k \leftarrow a_k)$ to mean that only the outer-substructure on $\{c_1, \ldots, c_k\}$ is being defined. 

We will always use the pull operation in such contexts that no conflicts with previously defined parts of $\str{C}$ will arise.
In particular, the $1$-types of the $c_i$'s will be always defined before any use of this pull operation, and further these $1$-types must agree, \emph{i.e.}, for all $i = 1,\dots,k$, $\type{\str{C}}{c_i}=\type{\str{A}}{a_i}$.
The following fact easily follows from the definition.

\begin{lem}\label{l:preserving}
Let $\str{A}$ and $\str{C}$ be structures, let $a_1, \ldots, a_k$ be a sequence of elements of $A$,
and let $c_1, \ldots, c_k$ be a sequence of distinct elements of $C$. If the substructure $\str{C} \restr \{c_1, \ldots, c_k\}$
was defined by $\str{C}\restr \{c_1, \ldots, c_k\} \leftarrow pull(\str{A}, c_1 \leftarrow a_1, \ldots, c_k \leftarrow a_k)$ 
(resp.~$\str{C}\restr_* \{c_1, \ldots, c_k\} \leftarrow pull^*(\str{A}, c_1 \leftarrow a_1, \ldots, c_k \leftarrow a_k)$),
then, for any first-order quantifier-free formula without equality $\psi(x_1, \ldots, x_l)$ (resp.~build out of
atoms which use all the variables $x_1, \ldots, x_l$ or just one of them), and any sequence of indices $i_1, \ldots, i_l$, we have $\str{A} \models \psi(a_{i_1}, \ldots, a_{i_l}) \text{ iff } \str{C} \models \psi(c_{i_1}, \ldots, c_{i_l})$. 
\end{lem}

\subsection{Normal form of AUF$_{\mathbf{1}}$.} \label{s:normalform}
We introduce a normal form of \AUF{}-sentences, inspired by Scott's normal form for \FOt{} (cf.~\cite{Sco62,GKV97}),
together with an efficient procedure for converting an arbitrary \AUF{}-sentence into equisatisfiable one being in this normal form.

We say that an \AUF{}-sentence is \emph{simple} if its shape is $\Qfr_1v_1 \ldots \Qfr_kv_k \psi(v_1,\ldots, v_k)$,
with the $\Qfr_i$'s being quantifiers, and $\psi$ being a quantifier-free formula with variables in $\{v_1,\dots,v_k\}$.
Additionally, we require that all the $v_i$'s are distinct. We remark that we permit here the empty sequence of quantifiers, with $\psi$ being a Boolean combination of $0$-ary predicates.
Finally, an \AUF{}-formula is in \emph{normal form} if it is a conjunction of the so-defined simple formulas.

Now, we describe our translation procedure.
The general idea behind it is to iteratively replace subformulas with fresh $0$- and $1$-ary symbols, and then axiomatise them.
More precisely, let $\phi$ be an \AUF{}-sentence.
We assume that it is in negation normal form, as we can use standard procedures for this.
We apply the following transformation step in loop, until $\phi$ reaches its normal form.
In each step, select $\psi$ as one of its innermost \AUF{}-subformulas beginning with a quantifier. There are two cases:
\begin{itemize}
  \item If $\psi$ has a free variable, \emph{i.e.}, its shape is $\Qfr_2 v_2 \ldots \Qfr_k v_k \psi_0(v_1,\ldots, v_k)$, with $\psi_0$ being quantifier-free,
add the conjunct, $\mu_\psi$, (partially) axiomatising a fresh unary symbol $P_\psi$:
\begin{align}\nonumber
  \forall v_1 \Qfr_2 v_2 \ldots \Qfr_k v_k \big[\neg P_\psi(v_1) \vee \psi_0(v_1, \ldots, v_k)\big],
\end{align}
and replace then $\psi$ by $P_\psi(v_1)$ in $\phi$; formally $\phi$ is now $\phi[P_\psi(v_1)/\psi] \wedge \mu_{\psi}$.
  \item Otherwise, $\psi$ is a proper subsentence, \emph{i.e.}, its shape is $\Qfr_1 v_1 \ldots \Qfr_k v_k \psi_0(v_1,\ldots, v_k)$, with $\psi_0$ being quantifier-free,
add the conjunct, $\nu_\psi$, (partially) axiomatising a fresh nullary symbol $E_\psi$:
\begin{align}\nonumber
  \Qfr_1v_1 \ldots \Qfr_kv_k \big[\neg E_\psi \vee \psi_0(v_1,\ldots, v_k)\big],
\end{align}
and replace then $\psi$ by $E_\psi$ in $\phi$; formally $\phi$ is now $\phi[E_\psi/\psi] \wedge \nu_{\psi}$.
\end{itemize}
\noindent %FIXED indent
It is readily verified that this procedure runs in polynomial time, and produces a normal form \AUF{}-sentence, possibly over a signature extended by fresh $0$- and $1$-ary symbols.
Moreover, one can verify that our reduction preserves both \AUFm{} and \AUFthree{}:
if the input formula is in \AUFm{} (\AUFthree{}), then output formula remains in \AUFm{} (\AUFthree{}), respectively.
It is slightly less trivial to demonstrate that our algorithm indeed produces an equisatisfiable sentence. This is proven in the following lemma.
\begin{lem}\label{l:nf}
  Let $\phi$ be an \AUF{}-sentence, and let $\phi_{\rm{NF}}$ be its normal form obtain via our algorithm.
  Then every model of $\phi$ can be expanded to a model of $\phi_{\rm{NF}}$, and every model of $\phi_{\rm{NF}}$ is a model of $\phi$.
\end{lem}
\begin{proof}
  Let $\phi_t$ be a formula obtained after $t$ steps of our procedure (then $\phi_0$ is just $\phi$).
  It is enough to show the equisatisfiability of $\phi_t$ and $\phi_{t+1}$ in each step $t$.
  As the two cases follow the same line of reasoning,
  we focus on the first one, \emph{i.e.}, the selected subformula $\psi$ of $\phi_t$ has a free variable $v$, and hence $\phi_{t+1}$ is $\phi_t[P_{\psi}(v)/\psi] \wedge \mu_{\psi}$.

  Suppose that $\str{A}$ is a model of $\phi_t$, then we obtain a model $\str{A}'$ of $\phi_{t+1}$ by interpreting the fresh unary symbol $P_\psi$ as the set $\{ a \in A : \str{A}\models\psi[a] \}$.
  Hence, obviously, $\str{A}'\models \mu_{\psi}$ and also $\str{A}'\models\phi_t[P_\psi(v_1)/\psi]$. 

  In the opposite direction, we assume that $\str{A}'$ is a model of $\phi_{t+1}$.
  It may happen that the subformula $\psi(v)$ is true in more points than $P_\psi(v)$, indeed, the added conjunct $\mu_\psi$ only entails $\psi(v)$, but not vice-versa.
  However, it suffices to guarantee that $\str{A}'\models\phi_t$, as $\phi_t$ is in negation normal form, and thus $\psi$ does not appear in the scope of the negation symbol.
\end{proof}
\noindent %FIXED indent
An artefact of our algorithm is that it may introduce fresh $0$-ary symbols, which are a relatively uncommon feature in \FO{}.
However, in the context of satisfiability, we can further simplify our normal form formulas by eliminating these inconvenient symbols:
it is enough to nondeterministically guess the truth values of $0$-ary predicates, replace them by $\top$ or $\bot$, in accordance with the guess, and then simplify the formulas accordingly.
Hence, we have the following corollary:
\begin{cor}\label{l:nf2}
  If the subfragment of \AUFm{} (\AUFthree{}) consisting of normal form sentences without $0$-ary relation symbols  possesses the finite, exponential model property, then the whole \AUFm{} (\AUFthree{}) also possesses it.
  Consequently, the \NExpTime{} upper bound on the complexity of satisfiability of normal form sentences is also transferred.
\end{cor}

\section{Diminished alternating uniform one-dimensional fragment} \label{s:aufmany}

In this section we work on \AUFm{}. 

We first note that the decidability of the satisfiability problem for \AUFm{} follows from the results  on Maslov class 
$\bar{\text{K}}$ (the dual of the Maslov class $\text{K}$).
Sentences of  $\bar{\text{K}}$ are often presented in prenex form. They are then of the following 
shape
\begin{eqnarray}
\label{eq:maslov}
\exists y_1 \ldots \exists y_k \forall x_1 \ldots \forall x_l \Qfr_1 z_1 \ldots \Qfr_l z_m \psi,
\end{eqnarray}
where the $\Qfr_i$ are quantifiers, $\psi$ is a quantifier-free formula without equality,  
and every atom of $\psi$ satisfies one of the following conditions: (i) it contains at most one $x_i$- or $z_i$-variable, (ii) it 
contains all the $x_i$-variables and no $z_i$-variables, or (iii) it contains an existentially quantified  variable $z_j$ and no $z_i$-variables
with $i>j$.

Decidability of the satisfiability problem for $\bar{\text{K}}$ (or, actually, the equivalent validity problem for $\text{K}$) 
was shown by Maslov \cite{Mas71} by his own inverse method. Later, a few resolution-based proofs were proposed. A proof 
by Hustadt and Schmidt \cite{HS99} works also for the class $\bar{\text{DK}}$ of conjunctions of $\bar{\text{K}}$-sentences.

Now, one easily observes that every \AUFm{}-normal form conjunct is a prenex $\bar{\text{K}}$-sentence.  Indeed, every conjunct containing only universal quantifiers
is of the form (\ref{eq:maslov}) with $k=m=0$, and its every atom  satisfies either condition (i) or (ii); every conjunct whose quantifier prefix
ends with the existential
quantifier is of the form (\ref{eq:maslov}) with $k=l=0$ and its every atom satisfies (i) or (iii).

Hence any normal form \AUFm{}-sentence  belongs to 
 $\bar{\text{DK}}$, and we can infer:
 
\begin{thm}
The satisfiability problem for \AUFm{} is decidable.
\end{thm}
\noindent %FIXED indent
In none of the works mentioned above the complexity of $\bar{\text{K}}$ has not been established. It has, however, eventually been
done in the recent work \cite{FKM24}. It turned out that $\bar{\text{K}}$ as well as the class $\bar{\text{DK}}$ of conjunctions of $\bar{\text{K}}$-sentences
have \NExpTime-complete satisfiability problem and finite, exponential model property. This gives that the satisfiability problem
for \AUFm{} is also \NExpTime-complete. 

The aim of this section is to present direct proofs of the exponential model property and \NExpTime-completeness of  \AUFm{}.
The constructions in \cite{FKM24} establishing the complexity of $\bar{\text{DK}}$ are very involved, and have to deal with a few problems which are
not present in the case of \AUFm{}, so we believe that the material in this section remains quite valuable.

\begin{thm} \label{t:fmp}
\AUFm{} has the exponential model property, hence its satisfiability (= finite satisfiability) problem is \NExpTime-complete.
\end{thm}
\noindent %FIXED indent
In the rest of this section we prove the above theorem. Actually, we propose two different arguments for the existence of small
finite models: the first one is constructive, the second---probabilistic. In both we look at satisfaction of \AUFm{} sentences
via the concept of \emph{satisfaction forests}. We start with the constructive variant, but the reader willing to see first
perhaps a bit simpler probabilistic proof may read Section~\ref{s:satfor}, and then jump directly to Section~\ref{s:probabilistic}.

 By Corollary \ref{l:nf2} we may restrict attention to normal form sentences without $0$-ary predicates.
For convenience we split the set of the  conjuncts of  a normal form sentence $\phi$ into those whose all quantifiers are universal and those which end with the existential 
quantifier:
\begin{eqnarray} \label{eq:normal}
&&\phi = \bigwedge_{1\le i \le \mse} \phie_i
\wedge \bigwedge_{1\le i \le \msu}  \phiu_i,
\end{eqnarray}
where   $\phie_i= \Qfr_1^i x_1 \Qfr_2^i x_2 \ldots \Qfr_{k_i{-}1}^i  x_{k_i{-}1} \exists  x_{k_i} \psie_i,$
$\phiu_i=\forall x_1 \ldots x_{l_i} \psiu_i,$
for  $\Qfr_j^i \in \{ \forall, \exists \}$, $\psie_i=\psie_i(x_1, x_2, \ldots, x_{k_i})$ and $\psiu_i=\psiu_i(x_1, \ldots, x_{l_i})$.
By $M$ we denote the maximal number of variables in a conjunct of $\phi$.

\subsection{Satisfaction forests} \label{s:satfor}

In this subsection we introduce \emph{satisfaction forests}, which are auxiliary structures (partially) describing some finite models of normal form 
\AUFm{} sentences.  
We first explain how to extract a satisfaction forest from a given model $\str{B}$ of a normal form sentence $\varphi$. Then we
formally define satisfaction forests and relate their existence to the existence of finite models of normal form sentences.

\subsubsection{Extracting a satisfaction forest from a model} 
\label{s:extraction}
Let $\varphi$ be a normal form $\AUFm$ sentence and let $\str{B}$ be its model. Assume that $\varphi$ is as in (\ref{eq:normal}). The satisfaction forest will be
a collection of labelled trees, one tree for each existential conjunct of $\varphi$, showing how this conjunct is satisfied in $\str{B}$. There will be two labelling functions: $\cL_1$ will assign elements from $B$ to tree nodes (with the exception of the root, which will be assigned the special empty label), and $\cL_2$ will assign outer-substructures of $\str{B}$ to the branches of the trees.
By a \emph{branch} we mean here
a sequence of elements $d_1, \ldots, d_{k_i}$ such that $d_1$ is a child of the root, $d_{k_i}$ is a leaf, and each $d_{i+1}$ is a child of $d_i$. An initial segment of a branch will be called a \emph{partial branch}.

It is helpful to think that a satisfaction forest just represents a winning strategy of
the existential player in the standard verification game for $\str{B}$ and $\phi$.

Consider a single existential conjunct 
$\phie_i=\Qfr_1^i x_1 \Qfr_2^i x_2 \ldots \Qfr_{k_i{-}1}^i  x_{k_i{-}1} \exists  x_{k_i} \psie_i$ of $\phi$.
Its satisfaction tree $\cT_i$ is built in the following process.

Start with the root and label it by the empty label. The root forms \emph{level} $0$ of the tree. \emph{Level} $j$ (or \emph{level} $x_j$), $0 < j \le k_i$ will correspond
to the quantifier $\Qfr^i_j$ and hence to the variable $x_j$. Assume level $j-1$ has been constructed, for $0 < j  \le k_i$. For each of its nodes $d$: 
\begin{itemize} \itemsep0pt
\item If 
$\Qfr^i_j=\forall$, then for each element $b \in B$ add a child $d'$ of $d$ to $\cT_i$, and set $\cL_1(d'):=b$. Nodes added in this step are called \emph{universal nodes}.
\item If 
$\Qfr^i_j=\exists$, then let $d_1, \ldots, d_{j-1}=d$ be the sequence of non-root nodes on the partial branch ending at $d$, ordered from the child of the root towards $d$. Choose in $\str{B}$ an element $b$ such that 
$$\str{B} \models \Qfr_{j+1}^i x_{j+1} \ldots \Qfr_{k_i-1}^i x_{k_i-1}\exists x_{k_i} \psie_i (\cL_1(d_1), \ldots, \cL_1(d_{j-1}), b, x_{j+1}, \ldots, x_{k_{i}}).$$
By the choice of the labels for $d_1, \ldots, d_{j-1}$ and the fact that $\str{B} \models \phie_i$, it is clear that such an element exists. If $j<k_i$, then we call $b$ an \emph{intermediate witness} for $\phie_i$, and if $j=k_i$ we call it a \emph{final witness} for $\phie_i$. Add a single child $d'$ of $d$ to $\cT_i$, and set $\cL_1(d')=b$. This node $d'$ is called an \emph{existential node}.
\end{itemize}\noindent %FIXED indent
For a branch $\flat$ of the above-defined tree, 
we denote by $Set(\flat)$ the set of labels of the non-root elements of $\flat$, by $Set^-(\flat)$ the set of labels of non-root and non-leaf elements
of $\flat$ and by $Seq(\flat)$ the sequence
of the labels of the non-root elements of $\flat$, ordered from the child of the root towards the leaf.

We now define  function $\cL_2$ assigning labels to the branches of the tree. For a branch  $\flat$, we set $\cL_2(\flat)$ to be the outer-structure $\str{B} \restr_* Set(\flat)$.

To declare some properties of satisfaction forests we need the following notions.
An outer-structure $\str{H}$  is $\phiu_i$-\emph{compatible},
if for every sequence $a_1, \ldots, a_{l_i}$ of elements of $H$ such that  $\{a_1, \ldots, a_{l_i}\}=H$ 
we have $\str{H} \models \psiu_i(a_1, \ldots, a_{l_i})$.
An outer-structure is $\phiu$-\emph{compatible} if it is $\phiu_i$-compatible for every conjunct 
$\phiu_i$.
Further, a set of $1$-types $\{\alpha_1, \ldots, \alpha_k \}$ is $\phiu$-\emph{compatible} if for any $1 \le m  \le M$ (recall that  $M$ is the maximal number of
variables in a conjunct of $\phi$), any set of  distinct
elements $H=\{a_1, \ldots, a_m \}$ and any assignment $f: \{a_1, \ldots, a_m\} \rightarrow \{\alpha_1, \ldots, \alpha_k \}$
one can build a $\phiu$-compatible outer-structure on $H$ in which, for every $1 \le i \le m$, the $1$-type of  $a_i$ is $f(a_i)$.

\smallskip
We now collect some properties of the tree $\cT_i$ and its labellings for $\phie_i$ constructed as above.
\begin{enumerate}[label=(T\arabic*),leftmargin=1.5cm,align=left]%FIXED changed to auto label and alignment. additional left margin can be added
\item for $1 \le j \le k_i$, and every node $d$ from level $j-1$:
\begin{enumerate}[align=left]
	\item if $\Qfr_j^i=\forall$, then $d$ has precisely $| B |$ children, labelled by distinct elements of $B$ (recall that each of these children is called a universal node)
	\item if $\Qfr_j^i=\exists$, then $d$ has precisely one child (recall that this child is called an existential node),
	\end{enumerate}
	\item for every branch $\flat \in \cT_i$, if $Seq(\flat)=(a_1, \ldots, a_{k_i})$, we have $\cL_2(\flat) \models \psie_i(a_1, \ldots, a_{k_i})$,
	\item for every pair of branches $\flat_1, \flat_2 \in \cT_i$, for every $a \in B$ such that $a \in Set(\flat_1)$ and $a \in Set(\flat_2)$
 the $1$-types of $a$ in $\cL_2(\flat_1)$ and in $\cL_2(\flat_2)$ are identical,
	\item for every pair of branches $\flat_1, \flat_2 \in \cT_i$ such that $Set(\flat_1)=Set(\flat_2)$, we have that $\cL_2(\flat_1)  = \cL_2(\flat_2)$,
	\item for every branch $\flat \in \cT_i$, $\cL_2(\flat)$ is $\phiu$-compatible.
	\end{enumerate}
\noindent %FIXED indent
If $\mathcal{T}_i$ meets the above properties it is called a \emph{satisfaction tree over} $B$ for $\phie_i$.

We then collect some properties of the whole sequence of trees  $\str{T}=(\cT_1, \ldots, \cT_{\mse})$ and their labellings constructed for $\varphi$ and $\str{B}$.
%FIXED Same as above
\begin{enumerate}[label=(F\arabic*),leftmargin=1.5cm,align=left]
\item for every $i$, $\cT_i$ is a satisfaction tree over $B$ for $\phie_i$,
\item for every pair of branches $\flat_1 \in \cT_i, \flat_2 \in \cT_j$, $i \not=j$, for every $a \in B$
such that $a \in Set(\flat_1)$ and $a \in Set(\flat_2)$, 
	the $1$-types of $a$ in $\cL_2(\flat_1)$ and in $\cL_2(\flat_2)$ are identical,
	\item for every pair of branches $\flat_1 \in \cT_i, \flat_2 \in \cT_j$, $i \not=j$ such that $Set(\flat_1)=Set(\flat_2)$, we have that $\cL_2(\flat_1)  = \cL_2(\flat_2)$,
	\item the set of all $1$-types appearing in the  outer-structures defined as labels of  the branches of the trees in $\str{T}$ is $\phiu$-compatible.

\end{enumerate}

\noindent %FIXED indent
Properties (T3), (T4), (F2) and (F3) will be sometimes called the \emph{(forest) consistency conditions}. 

\begin{clm}\label{c:conds}
The sequence of trees $\str{T}=(\cT_1, \ldots, \cT_{\mse})$ and the labelling functions $\cL_1, \cL_2$, constructed as above for the structure $\str{B}$ and the sentence $\varphi$ satisfies conditions (T1)-(T5) and (F1)-(F4).
\end{clm}
\begin{proof} 
It is not difficult to see that each of the $\cT_i$ satisfies (T1)-(T5),  and that the whole sequence satisfies (F1)-(F3). 
The only non-obvious point is (F4). Let us prove that it is true. Let $\alpha_1, \ldots, \alpha_k$ be the list of all
$1$-types appearing in the outer-structures assigned by $\cL_2$ to the branches in the whole forest $\str{T}$. Let $m$ be a natural number, $H=\{a_1, \ldots, a_m \}$ be a set of fresh distinct elements,
and $f: \{a_1, \ldots, a_m \} \rightarrow \{\alpha_1, \ldots, \alpha_k \}$ an assignment of $1$-types to these elements. 
We need to construct a
$\phiu$-compatible outer-structure on $H$ in which, for every $i$, the $1$-type of  $a_i$ is $f(a_i)$. 
For each $i$ choose an element $g(a_i) \in B$ such that $\tp^{\str{B}}(g(a_i)) = f(a_i)$; $g$ need not be an injective function. 

Let us  define the outer-structure $\str{H}$ on $H$ by first setting the $1$-type of each $a_i$ to be $f(a_i)$, and
then setting $\str{H} \restr_* \{a_1, \ldots, a_m\} \leftarrow \text{pull}^*(\str{B}, g(a_1), \ldots, g(a_m))$. 

We claim that so defined $\str{H}$ is $\phiu$-compatible. To see this take any conjunct $\phiu_i$
and any sequence of elements $c_1, \ldots, c_{l_i}$ such that $\{c_1, \ldots, c_{l_i}\}=H$, and
assume to the contrary that $\str{H} \not\models \psiu_i(c_1, \ldots, c_{l_i})$.
But then $\str{B} \not\models \psiu_i(g(c_1), \ldots, g(c_{l_i}))$, as the truth-values of the appropriate atoms appearing in $\psiu_i$
in the two considered structures coincide by our definition of $\str{H}$. Contradiction. 
\end{proof}

\subsubsection{Satisfaction forests and the existence of finite models} Let $\varphi$ be a normal form \AUFm{} sentence (we do not assume that a model of $\varphi$ is known). Formally, a \emph{satisfaction forest} for
$\varphi$ \emph{over a domain} $B$ is a sequence of
	trees $\str{T}=(\cT_1, \ldots, \cT_{\mse})$ together with  labelling functions $\cL_1$, $\cL_2$, assigning elements of $B$ to the nodes of the $\cT_i$ (with the exception of their roots to which the special empty label is assigned) and, respectively, outer-structures to their branches, such that each of the trees $\cT_i$ satisfies conditions
	(T1)-(T5), and the whole sequence satisfies conditions (F1)-(F4), with respect to $\cL_1$, $\cL_2$.

\begin{lem} \label{l:modelexists}
A normal form \AUFm{} sentence $\varphi$ has a model over a finite domain $B$ iff it has a satisfaction forest over $B$.
\end{lem}
\begin{proof}
Left-to-right implication is justified by the extraction of a satisfaction forest from a given model of $\varphi$
described in Section~\ref{s:extraction}, and in particular by Claim \ref{c:conds}.

In the opposite
direction assume that a satisfaction forest over a finite domain $B$ for $\varphi$ is given. We construct a model $\str{B}$ of $\varphi$
over the domain $B$. The construction is natural:

\smallskip\noindent
\emph{Step 1: $1$-types.} The $1$-type of an element $b \in B$ is defined as the $1$-type of $b$ in the structure
$\cL_2(\flat)$ for an arbitrarily chosen branch $\flat$, in an arbitrarily chosen tree $\cT_i$, for which $b \in Set(\flat)$.

\smallskip\noindent
\emph{Step 2: Witnesses.}
For every tree $\mathcal{T}_i$ and  its every branch $\flat$ of $\mathcal{T}_i$, define the outer-structure on $Set(\flat)$ 
in accordance with $\cL_2(\flat)$. 

\smallskip\noindent
\emph{Step 3: Completion.}
For any set of distinct elements $\{b_1, \ldots, b_k\}$ whose outer-structure is not yet defined, choose any
$\phiu$-compatible outer-structure which retains the already defined $1$-types of the $b_i$.

\smallskip
Properties (T3), (F2), (T4) and (F3) guarantee that Step 1 and Step 2 can be performed without conflicts,
and the existence of an appropriate outer-structure in Step 3 is guaranteed by (F4).

\smallskip
It remains to see that $\str{B} \models \varphi$. 
Consider any existential conjunct of $\phi$, that is a conjunct $\phie_i=\Qfr_1^i x_1 \Qfr_2^i x_2 \ldots \Qfr_{k_i{-}1}^i  x_{k_i{-}1} \exists  x_{k_i} \psie_i$.
The satisfaction tree $\cT_i$ witnesses that $\phie_i$ indeed holds: it describes all possible substitutions for universally quantified
variables, and shows how intermediate and final witnesses for existential quantifiers can be chosen.
Consider now any universal conjunct $\phiu_i=\forall x_1 \ldots x_{l_i} \psiu_i$ and
let $b_1, \ldots, b_{l_i}$ be any sequence of elements of $B$ (possibly with repetitions).
Let $H=\{b_1, \ldots, b_{l_i}\}$. The outer-structure on $H$ has been defined either in Step 2 or in Step 3.
In both cases we know that it is $\phiu$-compatible, in particular it is $\phiu_i$-compatible,
so $\str{H} \models \psi_i^\forall(b_1, \ldots, b_{l_i})$.  
\end{proof}

\subsection{From a model to a satisfaction forest over a small finite domain}

We are ready to present the main construction of this section in which we show that every satisfiable  \AUFm-sentence has a satisfaction forest over a small finite domain. 

Let $\str{A}$ be a (possibly infinite) model of a normal form sentence $\phi$ of the shape as in (\ref{eq:normal}). We show how to construct a satisfaction forest over a domain of size bounded exponentially
in $| \phi |$. By Lemma \ref{l:modelexists} this will guarantee that $\phi$ has a finite model over such a bounded domain.

\subsubsection{Domain}
Let $L$ be the number of $1$-types (over the signature of $\phi$) realized in $\str{A}$, and let these types be enumerated as $\alpha_0, \ldots, \alpha_{L-1}$. Recall that $M$ is the maximal number of variables in a conjunct of $\phi$.  
We define the domain $B$ to be $\{0, \ldots, 2M-1 \} \times \{0, \ldots, \mse-1 \} \times \{0, \ldots, (M-1)^{M-1}-1 \} \times \{0, \ldots, L-1\}$.   Note that $M$ and $\mse$ are bounded linearly  and $L$ is bounded exponentially in $| \phi | $, and hence $ | B |$ is indeed bounded exponentially in $| \phi |$.

For convenience, let us split $B$ into the sets $B_i= \{(i,*,*,*) \}$ (here and in the sequel $*$ will be sometimes used as a wildcard in the tuples denoting elements of the domain). We will sometimes call $B_i$ the $i$-th \emph{layer} of $B$. 

\subsubsection{Some simple combinatorics: Extension functions} During the construction of the satisfaction forest we will design a special strategy for assigning 
labels to the leaves. To this end we introduce  an auxiliary combinatorial tool, which we will call
\emph{extension functions}.

Let us recall a well known Hall's marriage theorem \cite{Hal35}.
A \emph{matching} in a bipartite graph $(G_1, G_2, E$) is a partial injective function $f:G_1 \rightarrow G_2$ such that if $f(a)=b$, then
$(a,b) \in E$. 
\begin{thm}[Hall]
Let $(G_1, G_2, E)$ be a bipartite graph. There exists a matching covering $G_1$ iff for any set $W \subseteq G_1$, the number of vertices of
$G_2$ incident to the edges emitted from $W$ is greater or equal to $|W|$.
\end{thm}
\noindent %FIXED indent
For a natural number $n$, let $[n]$ denote the set $\{0, \ldots, n-1 \}$, and, for $1 \le l \le n$ let $[n]^{l}$, denote the set of all subsets of $[n]$ of cardinality $l$.

\begin{lem} \label{l:comb}
For every $l, m \in \N$ such that $1 \le  l  <  m$, there exists a $1{-}1$ function $ext_l: [2m]^{l} \rightarrow [2m]^{l+1}$ such that for any $S \in [2m]^{l}$ we have that
	$S \subseteq ext_l(S)$.
\end{lem}
\begin{proof}
Consider the bipartite graph $ ([2m]^{l}, [2m]^{l+1}, E)$ such that $(S, S') \in E$ iff $S \subseteq S'$. To show that a desired  $ext_l$ exists it
suffices to show the existence of a matching covering entirely the set $[2m]^{l}$. To this end we apply Hall's marriage theorem.
In our graph every 
node from $[2m]^{l}$ has degree $2m-l$ (given an $l$-element subset of $[2m]$ it can be expanded to an $l+1$ subset just by adding to it 
precisely one of the remaining $2m - l $ elements), and every node from $[2m]^{l+1}$ has degree $l+1$ (to obtain an $l$-element subset of a $l+1$-subset
one just removes one of the elements of the latter). Take a subset $W$ of $[2m]^{l}$.
The nodes of this subset are incident to $|W| \cdot (2m-l)$ edges in total. Let us see that the number of nodes in $[2m]^{l+1}$ incident to a node from $W$ is greater than or equal
to $| W |$. Indeed, assume to the contrary that it is not. Then at most $| W | - 1$ nodes absorb $|W| \cdot (2m-l)$ edges emitted by $W$,
but this means that  $|W| \cdot (2m-l) \le (|W|-1) \cdot (l+1)$. Rearranging this inequality we get that 
$|W|(2m-2l-1)  +l +1 \le 0$. But using the assumption that $0 < l < m$ we have that $(2m-2l-1)>0$ and hence
$|W|(2m-2l-1)  +l +1$  must be greater than $0$, too. Contradiction. Thus
our graph satisfies the Hall's theorem assumptions which guarantee the existence of a matching from $[2m]^{l}$ to $[2m]^{l+1}$, covering entirely $[2m]^{l}$.
This matching can be taken as $ext_l$. 
\end{proof}
\noindent %FIXED indent
For every $1 \le l < M$ choose an extension function $ext_l:[2M]^l \rightarrow [2M]^{l+1}$ and let $ext = \bigcup_{l=1}^{M-1} ext_l$, that is, $ext$ is a function which takes a non-empty subset of $[2M]$ of size at most $M-1$ and returns a superset 
containing precisely one new element. Obviously $ext$ remains an injective function.

\subsubsection{Construction of a satisfaction forest} We now describe how to construct a satisfaction forest $\str{T}=(\cT_1, \ldots, \cT_{\mse})$ over the domain $B$ with labellings $\cL_1$, $\cL_2$ for $\phi$.
 It should be helpful to announce how we are going to take care of the  consistency conditions for the whole forest:
\begin{itemize} \itemsep0pt
\item Conditions (F2) and (T3): 
With every element  $a=(*,*,*,l) \in B$  we associate the $1$-type  $\alpha_l$. Whenever $a$ will be used as a label of a node in a satisfaction tree, then
its $1$-type in the outer-structure defined for any branch containing a node labelled with $a$ will be set to $\alpha_l$.
\item Conditions (F3) and (T4): for a a pair of distinct branches $\flat_1$, $\flat_2$ (either belonging to the same tree or to two different trees), we will simply have $Set(\flat_1) \not= Set(\flat_2)$. This condition will be ensured by an appropriate use of the extension function. Here it is important that the last quantifier in every $\phie_i$-conjunct is existential, and hence the last node of every branch in $\mathcal{T}_i$  is also existential, so we can freely choose its label from $B$.

\end{itemize}

\noindent %FIXED indent
Let us explain how to construct a single  $\phie_i$-satisfaction tree $\cT_i$ over $B$  and its labellings.
The general shape of $\cT_i$ is determined by $\phie_i$ and $B$: we know how many nodes we need, we know which of them
are existential, and which are universal, we know the labels of the universal nodes. It remains to assign labels to existential nodes (elements of $B$) and to branches (outer-structures on the set of elements formed from the labels of the nodes on a branch).

For each node of $\cT_i$, we will define its \emph{pattern element} in $\str{A}$. Pattern elements will be pointed to by an auxiliary function $\pat:\cT_i \rightarrow A$. We will choose $\pat(d)$, so that
its $1$-type is equal to type of $\cL_1(d)$. We emphasize that patterns are defined for the nodes of the tree and not for elements of $B$, in particular it is not harmful if two nodes belonging to two different branches, but having the same label have the same pattern element.

Consider a node $d_k$ and assume that all its non-root ancestors $d_1, \dots, d_{k-1}$ have the function $\pat$ and their labels already defined.
We proceed as follows:

\begin{itemize} \itemsep0pt
\item If $d_k$ is universal, then its label $\cL_1(d_k)$ is known
\begin{itemize} \itemsep0pt
	\item 
If $\cL_1(d_k)=\cL_1(d_j)$ for some $j<k$, then we set $\pat(d_k)=\pat(d_j)$. 
\item If the label $\cL_1(d_k)$ is not used by the ancestors of $d_k$, then choose as $\pat(d_k)$ an arbitrary element of $\str{A}$ of the $1$-type assigned to
$\cL_1(d_k)$.
(In particular, we may use an element which was used by one of the ancestors of $d$)
\end{itemize}
\item If $d_k$ is existential, then we need to define both $\cL_1(d_k)$ and $\pat(d_k)$. By our construction we have  that 
\begin{eqnarray*}
\str{A} \models \exists x_k \Qfr_{k+1}^i x_{k+1} \ldots \Qfr_{k_i-1}^i x_{k_i-1}\exists x_{k_i} \psie_i (\pat(d_1), \ldots, \pat(d_{k-1}),x_k, x_{k+1}, \ldots, x_{k_i}).
\end{eqnarray*}
\noindent %FIXED indent
We choose an element $w \in A$ witnessing the previous formula, i.e., an element such that 
$$\str{A} \models \Qfr_{k+1}^i x_{k+1} \ldots \Qfr_{k_i-1}^i x_{k_i-1}\exists x_{k_i} \psie_i (\pat(d_1), \ldots, \pat(d_{k-1}), w, x_{k+1}, \ldots, x_{k_{i}}),$$
 and set $\pat(d_k)=w$. To define the label of $d_k$ we consider two cases:

\begin{itemize} \itemsep0pt
\item
If $d_k$ is not a leaf, then:
\begin{itemize} \itemsep0pt
\item if $\pat(d_j)=w$ for some $j<k$, then set $\cL_1(d_k)=\cL_1(d_j)$
\item otherwise we choose as $\cL_1(d_k)$ an arbitrary element of $B$ which has assigned the $1$-type $\tp^{\str{A}}(w)$, not used by the ancestors of $d_k$ (there are many copies
of each $1$-type in $B$ so it is always possible). 
\end{itemize}

\item
If $d_k$ is a leaf, then let $\flat$ be the branch of $d_k$, and let  $S=\{j: \cL_1(d_l) \in B_j \mbox{ for some } l<k \}$. Of course, $|S| <k \le M$ so $ext(S)$ is defined. Let $s$ be the unique member of $ext(S) \setminus S$.
We take as $\cL_1(d_k)$ an element $(s, i, t, l) \in B_s$ where $l$ is such that $\alpha_l=\tp^{\str{A}}(w)$, and where $t$ is chosen so that none of the  branches $\flat'$ 
of the current tree for
which the labels have been already defined 
such that $Set^-(\flat')=Set^-(\flat)$ used $(s, i, t, l)$ as the label of its leaf. 
We indeed have  enough elements for this, since obviously $|Set^-(\flat)| \le M-1$ and
thus there are at most $(M-1)^{M-1}$ different branches, whose nodes from the first $M-1$ levels are
labelled by elements of $|Set^-(\flat)|$ (recall that there are $(M-1)^{M-1}$ possible choices for $t$).

\end{itemize}

\end{itemize}

\noindent %FIXED indent
Take now any branch $\flat$ of $\cT_i$. It remains to define the outer-structure $\cL_2(\flat)$. 
For any relational symbol $R$ of arity $m$, and any sequence $a_{i_1}, \ldots, a_{i_m}$ of elements of $Set(\flat)$ containing
all the elements of $Set(\flat)$, we set $R(a_{i_1}, \ldots, a_{i_m})$ to be true iff
$R(\pat(a_{i_1}), \ldots, \pat(a_{i_m}))$ is true in $\str{A}$. For every $a_j$, its $1$-type is set to be equal to the $1$-type of $\pat(a_j)$.
 This completes the definition of the outer-structure $\cL_2(\flat)$ on $Set(\flat)$. 
Note that the construction ensures that this outer-structure satisfies $\psi_i^{\exists}(Seq(\flat))$.

\subsubsection{Correctness}

Let us now see that the above-defined satisfaction forest indeed satisfies all the required conditions. 

\begin{itemize} \itemsep0pt
\item Conditions (T1), (T2) and (T3) should be clear.
\item For (T4) we show that  there is no pair of branches $\flat_1$, $\flat_2$ in a tree $\cT_i$ with $Set(\flat_1)=Set(\flat_2)$.
For a branch $\flat$ denote  $S_\flat=\{j: \cL_1(d) \in B_j \mbox{ for some } n \in \flat \}$. Take any pair $\flat_1$, $\flat_2$ of different branches from $\cT_i$.
If $S_{\flat_1} \not=S_{\flat_2}$
then obviously $Set(\flat_1)\not=Set(\flat_2)$. In the opposite case, due to the use of the function $ext$, the labels of the leaves of $\flat_1$ and $\flat_2$ belong
to the same layer $B_s$, and no other labels of nodes from both branches belong to $B_s$.
This means that if $Set^-(\flat_1)\not=Set^-(\flat_2)$, then still $Set(\flat_1)\not=Set(\flat_2)$.
On the other hand, if $Set^-(\flat_1)=Set^-(\flat_2)$, then by our construction the leaves of $\flat_1$ and $\flat_2$ are two different elements $b_1=(s, i, x, *)$ and $b_2=(s,i,y,*)$ from layer $B_s$, and $b_1 \in Set(\flat_1)$ but $b_1 \not\in Set(\flat2)$
and thus $Set(\flat_1) \not=Set(\flat_2)$.
\item To show that (T5) holds assume to the contrary that for some branch $\flat$, $\str{H}=\cL_2(\flat)$ is not $\phiu$-compatible; take $i$ for
which it is not $\phiu_i$-compatible. So, for some sequence $a_1, \ldots, a_{l_1}$ such that $\{a_1, \ldots, a_{l_1}\}=Set(\flat)=H$, we have
$\str{H} \not\models \psiu_i ( a_1, \ldots, a_{l_i})$. But then the definition of the outer-structure in $\cL_2(\flat)$ implies that $\str{A} \not\models \psiu_i(\pat(a_1), \ldots, \pat(a_{l_i}))$, that is $\str{A}$ violates $\phiu_i$. Contradiction.
   
\item Conditions (F1), (F2) should be clear.
\item For (F3) we show that  there  is no pair of branches $\flat_1$ in a tree $\cT_i$ and  $\flat_2$ in  $\cT_j$ for $i\not=j$, with $Set(\flat_1)=Set(\flat_2)$. We proceed as in the argument for (T4), in the last step observing that the leaves  $b_1$ and $b_2$ are different from each other since
$b_1=(s, i, *, *)$ and $b_2=(s, j, *,*)$, for some $i\not=j$.
\item For (F4) we reason precisely as in the reasoning for (F4) in the proof of  the Claim in Section~\ref{s:satfor} (we just replace the structure $\str{B}$
from this proof with the currently considered structure $\str{A}$). 
\end{itemize}

\subsubsection{Complexity of satisfiability}
To complete the proof of  Theorem \ref{t:fmp} we need to see that the satisfiability problem for \AUFm{} is \NExpTime-complete.

The lower bound is inherited from the lower bound for \FOt{} without equality \cite{Lew80}.
Let us turn to the upper bound.
By Corollary \ref{l:nf2}, it suffices to show how to decide satisfiability of a normal form sentence $\varphi$. As we have seen,  if $\varphi$ is satisfiable, then it has a model with exponentially
bounded domain. We guess some  natural description of such a model $\str{A}$. We note
that this description is also of exponential size with respect to $|\varphi|$: indeed, we need to describe some number (linearly bounded in $|\varphi|$)
of relations of arity at most $|\varphi|$, and it is straightforward, taking into consideration the size of the domain, that a description of a single such relation is at most exponential in $|\varphi|$.
A verification of a single normal form conjunct in the guessed structure can be done in an exhaustive way, by considering all possible substitutions
for the variables.

Alternatively, instead of guessing a model, one could guess a satisfaction forest for $\varphi$. Again, a routine inspection reveals that the size of its
description can be bounded exponentially in $|\varphi|$; also the verification of the properties (T1)-(T5), (F1)-(F4) would not be problematic.

\subsection{An alternative proof by the probabilistic method} \label{s:probabilistic}

In this subsection we reprove Theorem~\ref{t:fmp} by a use of the probabilistic method. In the context of the decision problem for fragments of
first-order logic this method was used first by  Gurevich and Shelah to show the finite model property for the G\"odel class without equality
\cite{GS83}. Our approach is inspired by this, though it is more complicated due to the specificity of our logic. We will also present an analysis of the size of a minimal model whose existence is guaranteed by our reasoning.

It will be convenient to employ the following notion. We say that a structure $\str{A}$ \emph{properly satisfies} a sentence $\phi$
if the existential player can win any play of the standard verification game using a strategy which requires her to always substitute fresh
elements (that is elements that have not yet been used in the play) for the existentially
quantified variables. Similarly a satisfaction forest is \emph{proper} if for every existential node its $\cL_1$-label does not label
any of its ancestors. 

It is easy to see that any satisfiable equality-free first-order sentence is properly satisfied in some model. For example,
if 
$\str{A} \models \phi$, then $\phi$ is properly satisfied in the structure  $\str{A}'$, defined as the structure over the domain $A \times \N$ such that
$\str{A}' \restr \{ (a_1, i_1), \ldots, (a_k, i_k) \} \leftarrow pull(\str{A}, (a_1, i_1) \leftarrow a_1, \ldots, (a_k, i_k)
\leftarrow a_k)$, for all $k$, the $a_j$ and the $i_j$. Indeed, for any $a \in A$, the elements $(a,0), (a,1), \ldots$, are then indistinguishable
and the existential player has then always many ``proper'' choices.
Similarly, any satisfiable
\AUFm{} sentence has a proper satisfaction forest that can be naturally extracted from a proper model.

\subsubsection{Compressed satisfaction forest.}

Let $\phi$ be a satisfiable normal form \AUFm{}-sentence as in (\ref{eq:normal}), and let $\str{A} \models \phi$ be its arbitrary (possibly infinite) proper  model.

Let  $\str{T}=(\cT_1, \ldots, \cT_{\mse})$ with labelling functions $\cL_1, \cL_2$ be a proper satisfaction forest for $\phi$ over the domain $A$. 
Recalling that $M$ is the maximal number of variables in a conjunct of $\phi$, let $\AAA=\{\type{\str{A}}{a}: a \in A\}$ be the set of $1$-types realized in $\str{A}$ and
$\BBB$ the set of all relevant outer-types realized in $\str{A}$:
\begin{eqnarray} 
\nonumber && \BBB=\{\outertype{\str{A}}{a_1, \ldots, a_k}: 1 \le k \le M, a_1, \ldots a_k \text{ distinct elements of } A \}
\end{eqnarray}

We first rebuild $\str{T}$ into a \emph{compressed} satisfaction forest $\str{T}^*=(\cT_1^*, \ldots, \cT^*_{\mse})$, with labellings $\cL^*_1, \cL^*_2$ in which the labelling function for nodes returns elements
of $\AAA \cup \Vars$ rather than elements of $A$ and the labelling function for branches returns outer-types rather then outer-structures
on subsets of the domain of $A$. The intuition is as follows. If the $\cL_1^*$-label of a node $d$ from a level $x$ is an $\alpha \in \AAA$, then $d$ represents the assignment (during a play of the verification game) of any element of $1$-type $\alpha$ for $x$, while if it is an $y \in \Vars$, then $d$ represents 
assigning for $x$ the element which is already assigned for $y$.

In the first step, we define the new labelling of nodes, $\cL_1^*$. For every $\cT_i$ and its every node $d$:
\begin{itemize}
\item if there is an ancestor of  $d$ with the same label as $d$, then take such an ancestor $d'$ which is closest to
the root; let its level be the level $x_i$; set $\cL_1^*(d)=x_i$. 
\item if there is no such an ancestor, then set $\cL_1^*(d)=\tp^{\str{A}}(\cL_1(d))$. 
\end{itemize}
\noindent %FIXED indent
Next we cut some nodes of the $\cT_i$ leaving only its sufficient ``compressed'' representation.
In the depth-first manner, for every tree $\cT_i$ and its every node $d$, if the children of $d$ belong to 
a universal level $x_i$, then, for every $\alpha \in \AAA$, mark one of the children $d'$ of $d$ such that $\tp^{\str{A}}(\cL_1(d'))=\alpha$,
and mark all the children $d'$ of $d$ for which $\cL_1^*(d') \in \Vars$. 
After this, remove all the subtrees rooted at unmarked children of $d$. 
Denote $\cT_i^*$ the trees so-obtained. 

Finally, we define the labelling function $\cL^*_2$.  For each branch $\flat$ that survived the above process, let $Seq^*(\flat)$ denotes
the subsequence of $Seq(\flat)$ in which only the first occurrence of every element is retained.
As $\cL_2^*(\flat)$ we take $\outertype{\str{A}}{Seq^*(\flat)}$.

 Let $\BBB_{wit}:=\{\cL_2^*(\flat): \flat \text{ is a branch of } \cT^*_i \text{ for some } i\}$,
and let $\BBB_{wit}^+$ be the closure of $\BBB_{wit}$  under permutations of the domains of outer-types: $\BBB_{wit}^+:=\{ \beta \in \BBB: \text{there is } \beta'\in \BBB_{wit} \text{ isomorphic to } \beta \}$.
For every $1 \le k \le M$ and every sequence of $1$-types $\alpha_1, \ldots, \alpha_k$, choose a $k$-outer-type in which
the $1$-type of the element $i$ is $\alpha_i$. Let $\BBB_{cmp}$ be the set of the so-chosen outer-types.
Let $\BBB^*=\BBB_{wit} \cup \BBB_{cmp}$. 

Note that the degree of nodes of $\cT_i^*$ is bounded by $|\AAA \cup \Vars|$
and the number of branches by $(|\AAA \cup \Vars|)^M$. Hence the number of outer-types 
in $|\BBB^*|$ is not greater than $M! \cdot (|\AAA \cup \Vars|)^M + M |\AAA|^M$,
which is exponential in $|\phi|$.

We say that a (partial) branch $\flat=(d_1, \ldots, d_k)$ of $\cT_i^*$  \emph{corresponds} to a sequence of elements $a_1, \ldots, a_k$ of $A$
if for all $i=1, \ldots, k$:
\begin{itemize}
\item whenever $a_i \not= a_j$ for all $j<i$, then $\cL_1^*(d_i)=\type{\str{A}}{a_i},$
\item if $a_i=a_j$ for some $j<i$, then, denoting by $j_0$ the smallest such $j$,  $\cL_1^*(d_i)=x_{j_0}.$
\end{itemize}
Of course, a branch may correspond to many different sequences of elements.

\subsubsection{Random structures.}
Let us enumerate the $1$-types in $\AAA$ as $\alpha_0, \ldots, \alpha_{L-1}$.
For every $n \in \N$, we generate a random structure $\str{A}_s$ in the following steps:

\smallskip\noindent
\emph{Stage 0: The Domain.} Set $A_n=\{0, \ldots, L-1\} \times \{1, \ldots, n\}$.

\smallskip\noindent
\emph{Stage 1: The $1$-types}. For every $(i,j)$  set $\type{\str{A}_n}{(i,j)}:=\alpha_i$.

\smallskip\noindent
\emph{Stage 2: The $k$-outer-types for $2 \le k \le M$}. For every subset of $A$ of size at most $M$, take an arbitrary enumeration of its elements $a_1, \ldots, a_k$, and let $\alpha_1, \ldots, \alpha_k$ be the sequence of their $1$-types in $\str{A}$. Let $\BBB_0$ be the set of $k$-outer-types $\beta$ from
$\BBB^*$ for which $\beta\restr i = \alpha_i$, for $i=1, \ldots, k$. $\BBB_0$ is not empty by the construction of $\BBB_{cmp}$. Choose $\outertype{\str{A}_n}{a_1, \ldots, a_k}$ randomly
from $\BBB_0$, with uniform probability.

\smallskip\noindent
\emph{Stage 3: The $k$-outer-types for $k>M$}. For every subset of $A$ of size bigger than $M$, take an enumeration $a_1, \ldots, a_k$ of its elements,
and set $\outertype{\str{A}_n}{a_1, \ldots, a_k}$ arbitrarily, e.g. by setting all the non-unary facts to be false.

\subsubsection{Probability estimation}
Let us now estimate the probability $\Prob[\str{A}_n \not\models \phi]$. As all the outer-types in $\BBB^*$ are taken from a model of $\phi$ they
are all $\phiu$-compatible, so we have that $\str{A}_n \models \phiu$.
Hence:

$$\Prob [\str{A}_n \not\models \phi] = \Prob[\bigcup_{i=1}^{\mse} \str{A}_n \not\models \phie_i] \le \sum_{i=1}^{\mse} \Prob[ \str{A}_n \not\models \phie_i].$$

Consider a single conjunct $\phie_i=\Qfr_1^i x_1 \Qfr_2^i x_2 \ldots \Qfr_{k_i{-}1}^i  x_{k_i{-}1} \exists  x_{k_i} \psie_i$.
We estimate $\Prob[ \str{A} \not\models \phie_i]$ by the probability that this formula is not satisfied under the fixed strategy of choosing 
intermediate witnesses for existential quantification, recorded in $\cT^*_i$. 
Formally, this strategy works as follows. Let $l$ be the number of universal quantifiers in $\phie_i$, and assume that elements 
$a_1, \ldots, a_l$ are going to be assigned to universally quantified variables. 
Denote $\bar{a}=(a_1, \ldots, a_l)$. We define the sequence of elements $f(\bar{a})=\bar{b}=(b_1, \ldots, b_{k_i})$
that should be substituted for all variables of $\phie_i$. We proceed from left to right. Assume $b_1, \ldots, b_{j-1}$ are
defined. If the $j$-th quantifier $Q^i_j$ of $\phie_i$ is universal, then let $k$ be the number of
universal quantifiers among $Q^i_1, \ldots, Q^i_{j}$. Set $b_j:=a_k$. If $Q^{i}_j$ is existential,
then let $\flat$ be the partial branch of $\cT^*_i$ corresponding to the sequence $(b_1, \ldots, b_{j-1})$.
Let $d$ be the successor of the last element of $\flat$ (it is unique, as it belongs to an existential level), and let $\alpha_u=\cL^*_1(d)$. Set $b_j:=(u,t)$,
where $t$ is the smallest index such that $(u,t) \not\in \{b_1, \ldots, b_{j-1}\}$.

We have

$$\Prob [\str{A}_n \not\models \phie_i] \le \Prob[\bigcup_{\bar{a} \in A_n^l} \str{A}_n \not\models \exists x_{k_i} \psie_i(f(\bar{a}), x_{k_i})],$$

and further

$$ \Prob[\bigcup_{\bar{a} \in A_n^l} \str{A}_n \not\models \exists x_{k_i} \psie_i(f(\bar{a}), x_{k_i})] \le
\sum_{\bar{a} \in A_n^l}  \Prob[\str{A}_n \not\models \exists x_{k_i} \psie_i(f(\bar{a}), x_{k_i})]$$

We need to estimate $\Prob[\str{A}_n \not\models \exists x_{k_i} \psie_i(f(\bar{a}), x_{k_i})]$. Let us consider the partial branch of
$\cT^*_i$ corresponding to $f(\bar{a})$. Let $d$ be the element in the last existential level completing this branch, and denote the
whole branch by $\flat$.
Let $\alpha_u=\cL_1^*(d)$.  
There are at least $n-k_i+1$ elements of $1$-type $\alpha_u$ is $\str{A}_n$ that do not appear in $f(\bar{a})$. Consider a single such 
element $b$. In the process of defining the outer-types in $\str{A}_n$ the outer-type of some tuple being a permutation of  $(f(\bar{a}), b)$ was chosen randomly from $\BBB^*$. One possible candidate for this outer-type was the corresponding permutation of the outer-type $\cL_2^*(\flat)$. If it was this particular outer-type that was chosen, then $\str{A}_n \models \psie_i(f(\bar{a}), b)$. The chance for this was at least $1 / |\BBB^*|$. Hence

$$ \Prob [\str{A}_n \not\models \psie_i(f(\bar{a}), b)] \le 1 - 1/|\BBB^*|.$$

For distinct candidate final witnesses  $b_j$ as above, the events $[\str{A}_n \not\models \psie_i(f(\bar{a}), b_j)]$ are independent (note that this is the place in our reasoning where we use the assumption that the considered quantifier prefix ends with the existential quantifier), which implies:

$$\Prob[\str{A}_n \not\models \exists x_{k_i} \psie_i(f(\bar{a}), x_{k_i})] \le (1-1/|\BBB^*|)^{n-k_i+1}.$$

Then, noting that the number of tuples $\bar{a} \in A_n^l$ is $(nL)^l$:

$$\Prob [\str{A}_n \not\models \phie_i)] \le (nL)^l \cdot (1-1/|\BBB^*|)^{n-k_i+1},$$

and finally 

\begin{equation}
  \label{eqn:bound-failure}
  \Prob [\str{A}_n \not\models \phi] \le \mse \cdot (nL)^M \cdot (1-1/|\BBB^*|)^{n-M+1}.
\end{equation}

As the parameters $\mse$, $L$, $M$ and $|\BBB^*|$ are independent from $n$ the limit of the above estimation
as $n \rightarrow \infty$ is $0$. So, for sufficiently large $n$, it is smaller than $1$. That is for
some $n$ we have that $\Prob[\str{A}_n \models \phi]>0$. This means that there is a finite model of $\phi$. 

\subsubsection{Size of models.}
We can estimate the size of models obtained as above via the probabilistic method by further analysing the last
estimation on $\Prob [\str{A}_n \not\models \phi]$.
More precisely, we need to find a bound on $n$ (as a function of $|\phi|$) for which $\Prob [\str{A}_n \not\models \phi]<1$,
and hence conclude the existence of model $\str{A}_n$ of size $nL$ (as the domain of $\str{A}_n$ is $\{1,\dots,L\}\times\{1,\dots,n\}$).
The following claim establishes such a bound.

\begin{clm}
  There exists an exponential function $f(x)$ such that $\Prob[\str{A}_n \not\models\phi]<1$ for all $n\ge f(|\phi|)$.
\end{clm}
\begin{proof}
First, recall that we already established that $|\BBB^*|$ is at most exponential in $|\phi|$, hence there exists a polynomial $p(x)$ such that $|\BBB^*|<2^{p(|\phi|)}$.
Similarly, by Lemma~\ref{l:ntyp} the number of 1-types $L$ can be bounded by $2^{q(|\phi|)}$ for some polynomial $q(x)$.
We may assume that $p(x)$ and $q(x)$ are monotonic and at least $10x$.
Finally, observe that $\mse$ and $M$ are at most $|\phi|$.

Hence, we can rewrite the right-hand side of (\ref{eqn:bound-failure}) in terms of $p(|\phi|),q(|\phi|)$ and $|\phi|$:
$$ \mse \cdot (nL)^M \cdot (1-1/|\BBB^*|)^{n-M+1} \le |\phi|\cdot (n2^{q(|\phi|)})^{|\phi|} \cdot (1-1/2^{p(|\phi|)})^{n-|\phi|+1}. $$

Now, we choose a function $f(x)$ such that for all $n \ge f(|\phi|)$ the right side of the inequality above is strictly less than $1$.
Let $f(x)=20 \cdot 2^{2\cdot p(x)}\cdot q(x) \cdot \ln(2) \cdot x^2$, which is clearly an exponential function.
Hence, to conclude we need to show that $$ x\cdot (f(x)\cdot2^{q(x)})^{x} \cdot (1-1/2^{p(x)})^{f(x)-x+1} < 1. $$

We use the well-known inequality $1-y \le e^{-y}$ to obtain that
$$ x\cdot (f(x)\cdot2^{q(x)})^{x} \cdot (1-1/2^{p(x)})^{f(x)-x+1} \le x\cdot (f(x)\cdot2^{q(x)})^{x} \cdot \exp(-(f(x)-x+1)/2^{p(x)}) < 1. $$
Then we apply natural logarithm to simplify:
$$ \ln(x) + x\cdot \ln f(x) + x\cdot q(x)\cdot\ln(2) -(f(x)-x+1)/2^{p(x)} < 0,$$
and further we rearrange the terms:
$$  2^{p(x)}\cdot \ln(x) + 2^{p(x)}\cdot x\cdot \ln f(x) + 2^{p(x)}\cdot x\cdot q(x)\cdot\ln(2) + x{-}1 < f(x).$$

We can see that each of $2^{p(x)}\cdot \ln(x)$,  $2^{p(x)}\cdot x\cdot q(x)\cdot\ln(2)$ and $x{-}1$ is clearly less than $f(x)/5$.
Thus, it is left to show that $2^{p(x)}\cdot x\cdot \ln f(x)$ is also less than $f(x)/5$.
However, this is not difficult to see either, as $2^{p(x)}\cdot x < \sqrt{f(x)}$ by the definition of $f(x)$ and also $\ln y < \sqrt{y}/5$ for all $y\ge 1280$.
Therefore, as $f(x)$ is monotonic, and already $f(1)\ge1280$, we have that $$2^{p(x)}\cdot x\cdot \ln f(x) < \sqrt{f(x)}\cdot \sqrt{f(x)}/5 = f(x)/5,  $$
which finishes the proof.
\end{proof}

\section{The three-variable alternating uniform one-dimensional fragment}
\label{s:auf}
The aim of this section is to show the following result for \AUFthree{}.

\begin{thm} \label{t:fmpauf}
\AUFthree{}  has the exponential model property.
Hence the satisfiability problem (=finite satisfiability problem) for \AUFthree{} is \NExpTime-complete.
\end{thm}

Clearly, the upper bound in the second part of this theorem  is implied by the first part: to verify that a given \AUFthree{}
sentence is satisfiable we just guess a description of its appropriately bounded model, and verify it (in a straightforward manner).
We recall that the corresponding lower bound is retained from \FOt{} without equality \cite{Lew80}.

Recall that by Corollary \ref{l:nf2} we can restrict attention to normal form formulas without $0$-ary predicates. For convenience, we eliminate conjuncts whose quantifier prefix starts with the
existential quantifier, like $\exists x \Qfr y \Qfr z \eta_i$, by replacing them by $\forall x \exists y P(y) \wedge \forall x \Qfr y \Qfr z (P(x) \rightarrow \eta_i)$, for a fresh symbol $P$. We also replace  conjuncts $\forall x \psi(x)$ whose quantifier prefix consists just of a single universal quantifier,
by $\forall x \forall y \psi(x)$.  

Hence, a normal form  \AUFthree{} sentence can be written as 
$\phi=\phi_1 \wedge \phi_2$ for
\begin{align*}
\phi_1=&  \bigwedge_{i \in \mathcal{I}_1} \forall x \forall y \phi_i(x,y) \wedge \bigwedge_{i \in \mathcal{I}_2} \forall x \forall y \forall z \phi_i(x,y,z) \wedge 
\bigwedge_{i \in \mathcal{I}_3} \forall x \exists y \forall z \phi_i(x,y,z), \\
\phi_2=& \bigwedge_{i \in \mathcal{I}_4} \forall x \exists y  \phi_i(x,y) \wedge \bigwedge_{i \in \mathcal{I}_5} \forall x \forall y \exists z \phi_i(x,y,z) \wedge \bigwedge_{i \in \mathcal{I}_6} \forall x \exists y \exists z \phi_i(x,y,z).
\end{align*}
where the $\mathcal{I}_i$ are pairwise disjoint sets of indices, and all the $\phi_i$ are {uniform} Boolean combinations of atoms over the appropriate sets of variables
($\{x,y,z \}$ or $\{x,y \}$), that is of atoms containing either all of the free variables of $\phi_i$ or just one of them.

Conjuncts indexed by $\mathcal{I}_i$, $1 \le i \le 6$, will be sometimes called the $\mathcal{I}_i$-conjuncts.

For the rest of this section we fix an \AUFthree{} normal form sentence $\phi=\phi_1 \wedge \phi_2$ and its model $\str{A}$.
  We show how to construct
a small model of $\phi$. Compared to small model constructions for other fragments of first-order logic, one non-typical 
problem we need to deal with here is respecting $\mathcal{I}_3$-conjuncts, that is conjuncts starting with the quantifier pattern $\forall\exists\forall$. To help the reader understand what happens, we split our
construction into two parts. We first build (Subsection~\ref{s:phi1}) a small model $\str{B} \models \phi_1$; here we indeed need to concentrate mostly on the $\mathcal{I}_3$-conjuncts.
Then (Subsection~\ref{s:phi2}) we will use $\str{B}$ as a building block in the construction of a finite, exponentially bounded, model for the whole $\phi$. By Corollary \ref{l:nf2}, this suffices to prove Theorem \ref{t:fmpauf}.

\subsection{Dealing with  \texorpdfstring{${\forall\exists\forall}$}{}-conjuncts.} \label{s:phi1}

Here, we give a construction of a model $\str{B}$ satisfying all the ($\mathcal{I}_1 \cup \mathcal{I}_2 \cup \mathcal{I}_3$)-conjuncts (i.e. $\str{B} \models \phi_1$).
To properly handle the $\mathcal{I}_3$-conjuncts, we pick for every $j \in \mathcal{I}_3$ a witness function $\ff_j:A \rightarrow A$
such that, for every $a \in A$, we have $\str{A} \models \forall z \phi_j(a, \ff_j(a), z)$. An element $\ff_j(a)$ will be sometimes
called a \emph{witness} for $a$ and $\phi_j$ (or for $a$ and the appropriate $\mathcal{I}_3$-conjunct). Similar terminology
will be also used later for conjuncts from $\phi_2$.
Here and in the rest of this section we will use the convention that when referring to the indexes of $\mathcal{I}_3$-conjuncts,
we will use the letter $j$, while for the other groups of conjuncts (or in situations where it is not relevant which group of
conjuncts is considered) we will usually use the letter $i$. 

Let us introduce the following definition of a \emph{generalized type} of an element $a \in A$.
Formally, denoting $\AAA$ the set of $1$-types realized in $\str{A}$, a generalized type is the following tuple:
$$ \tau = (\alpha(\tau), \beta(\tau), j(\tau), s(\tau)) \in \AAA \times (\AAA \cup \{ \bot \}) \times (\{1, \ldots |\mathcal{I}_3| \} \cup \{ \bot\}) \times \AAA^{|\mathcal{I}_3|}. $$
We say that an element $a \in A$ realizes a generalized type $\tau$ if:
\begin{itemize}\itemsep0pt
  \item $\alpha(\tau)$ is the 1-type of $a$;
  \item if $\beta(\tau)\neq\bot$, then there exists an element $b \in A$ such that $\beta(\tau)$ is the 1-type of $b$, and $a$ is the witness for $b$, and the $\mathcal{I}_3$-conjunct with index $j(\tau) \neq \bot$, i.e. $f_{j(\tau)}(b) = a$;
    and if $\beta(\tau)=\bot$, then we do not impose any requirements, except that $j(\tau)=\bot$;
  \item $s(\tau) = (s(\tau)_1, \ldots, s(\tau)_{|\mathcal{I}_3|})$ is the sequence of 1-types of witnesses for $a$ for all the $\mathcal{I}_3$-conjuncts, i.e., $s(a)_j = \type{\str{A}}{f_j(a)}$ for all $1 \le j \le |\mathcal{I}_3|$.
\end{itemize}
Intuitively, if $a$ realizes $\tau$, then $\tau$ captures the information about the 1-types of $a$ itself and all the witnesses of $a$,
from the perspective of an element $b$ witnessed by $a$ (if such an element $b$ exists).
We want to stress that one element might realize potentially \emph{many} generalized types, since the choices of $\beta(\tau)$ and $j(\tau)$ might not be unique. The set of generalized types realized by $a$ in $\str{A}$ is denoted $\gtp(a)$. 

Let $\gtp(\str{A}) := \bigcup_{a \in A} \gtp(a)$ be the set of all  generalized types realized in $\str{A}$.
For each $\tau \in \gtp(\str{A})$ pick a single representative $\pat_\gtp(\tau) = a_\tau$ such that $\tau \in \gtp(a)$  (we do not require $\pat_\gtp$ to be an injective function).

\smallskip\noindent\emph{Stage 0: The domain.} Let the domain of $\str{B}$ be the set $$ B = \gtp(\str{A}) \times \{0, 1, 2, 3 \}. $$
We will sometimes call the sets $B_m = \{ (\tau,m) : \tau \in \gtp(\str{A}) \}$ \emph{layers}.
We can naturally lift
the function $\pat_\gtp$ to the function $\pat:B \to A$: $\pat((\tau, m)) = \pat_\gtp(\tau)$.

Before proceeding to the construction, we need one more definition, namely the definition of a \emph{witnessing graph}. A witnessing graph will contain a \emph{declaration} about who is a witness for whom for $\mathcal{I}_3$-conjuncts in the final structure $\str{B}$.

Formally, we first define the \emph{pre-witnessing graph} to be the directed graph $\mathcal{G}'_w = (B,E')$ with the edge-labelling function $\ell: E' \rightarrow \{1,\ldots,|\mathcal{I}_3|\}$. In $\mathcal{G}'_w$ we connect
$b_1 = (\tau_1, m_1)$ to $b_2 = (\tau_2, m_2)$ by an arc $e$ with label $\ell(e) = j(\tau_2)$ if $m_1+1 = m_2 (\mod 4)$, $\beta(\tau_2) = \alpha(\tau_1)$, and $s(\tau_1)_{j(\tau_2)} = \alpha(\tau_2)$.
As the reader may notice:
\begin{itemize}
\item the pre-witnessing graph connects only elements which are in consecutive layers (modulo 4), 
\item for every vertex there is at least one outgoing arc in every colour, and 
\item there are no multiple arcs between the same pair of elements.
\end{itemize}
\noindent %FIXED indent
We next ``functionalise'' the pre-witnessing graph by restricting the arc-set $E'$ to contain exactly one outgoing arc in each colour for every vertex. Let this new arc-set be $E$.
Thus, we obtain a \emph{witnessing graph} $\mathcal{G}_w = (B,E)$ and write $b_1 \leadsto_j b_2$ if $b_1$ is connected to $b_2$ with an arc coloured by $j$.

The construction itself will be given by the following iterative algorithm.
The role of the generalized types and the witnessing graph is to guide the algorithm in selecting patterns from $\str{A}$ to be pulled into $\str{B}$, avoiding potential conflicts between different conjuncts (especially) of type $\mathcal{I}_3$.
For the clarity of the presentation, all the conditions present in the algorithm need to be taken modulo permutation of the considered elements in each step, and by $\ff_j^{-1}(a)$ we mean to take any $b \in A$ satisfying $\ff_j(b)=a$.
The algorithm is as follows:

\smallskip\noindent
{\em Stage 1: 1-types.} We set the $1$-types of elements $b$ in $\str{B}$, as given by $\alpha(b)$.

\smallskip\noindent
{\em Stage 2: 2-outer-types.} For all pairs $\{ b_1=(\tau_1,m_1),b_2=(\tau_2,m_2) \}$ of elements from $B$, $b_1 \not= b_2$ do:
  \begin{itemize} \itemsep0pt
    \item if $\outertype{\str{B}}{b_1, b_2}$ is already defined, then continue.
    \item if $b_1 \leadsto_j b_2$ for some $j$, then set $(\str{B} \restr \{ b_1, b_2\}) \leftarrow pull(\str{A}, b_1 \leftarrow a_{\tau_1}, b_2 \leftarrow \ff_j(a_{\tau_1}) )$.
    \item otherwise, set $(\str{B} \restr \{b_1, b_2\}) \leftarrow pull(\str{A}, b_1 \leftarrow a_{\tau_1}, b_2 \leftarrow a_{\tau_2} )$.
  \end{itemize}

\smallskip\noindent
{\em Stage 3: 3-outer-types.} For all triples $\{ b_1=(\tau_1,m_1),b_2=(\tau_2,m_2),b_3=(\tau_3,m_3) \}$ of pairwise distinct elements from $B$
do:
  \begin{itemize} \itemsep0pt %FIXED adjusted for uniformity
    \item if $\outertype{\str{B}}{b_1,b_2,b_3}$ is already defined, then continue.
    \item if $b_1 \leadsto_j b_2$ and $b_2 \leadsto_k b_3$, for some $j,k$, then set $$\str{B} \restr_* \{b_1, b_2, b_3\} \leftarrow pull^*(\str{A}, b_1 \leftarrow \ff_j^{-1}(a_{\tau_2}), b_2 \leftarrow a_{\tau_2}, b_3 \leftarrow \ff_k(a_{\tau_2})).$$

    \item if $b_1 \leadsto_j b_2$ and $b_1 \leadsto_k b_3$, for some $j,k$, then set $$\str{B} \restr_* \{b_1, b_2, b_3\} \leftarrow pull^*(\str{A}, b_1 \leftarrow a_{\tau_1}, b_2 \leftarrow \ff_j(a_{\tau_1}), b_3 \leftarrow \ff_k(a_{\tau_1})).$$

    \item if $b_2 \leadsto_j b_1$ and $b_3 \leadsto_k b_1$, then set $$\str{B} \restr_* \{b_1, b_2, b_3\} \leftarrow pull^*(\str{A}, b_1 \leftarrow a_{\tau_1}, b_2 \leftarrow \ff_j^{-1}(a_{\tau_1}), b_3 \leftarrow \ff_k^{-1}(a_{\tau_1})).$$

    \item if $b_1 \leadsto_j b_2$ for some $j$ and $b_3$ is connected to neither $b_1$ nor $b_2$ in the witnessing graph, then set $$\str{B} \restr_* \{b_1, b_2, b_3\} \leftarrow pull^*(\str{A}, b_1 \leftarrow a_{\tau_1}, b_2 \leftarrow \ff_j(a_{\tau_1}), b_3 \leftarrow a_{\tau_3}).$$

    \item otherwise, set $$\str{B} \restr_* \{b_1, b_2, b_3\} \leftarrow pull^*(\str{A}, b_1 \leftarrow a_{\tau_1}, b_2 \leftarrow a_{\tau_2}, b_3 \leftarrow a_{\tau_3} )$$.
  \end{itemize}

\begin{clm}
  $\str{B} \models \phi_1$.
\end{clm}
\begin{proof}
  Most of the proof follows from the same case analysis as in the algorithm itself. Hence, we give just a sketch and leave some details for the reader.

  First notice that in the final structure $\str{B}$ all the $\mathcal{I}_1$- and $\mathcal{I}_2$-conjuncts are satisfied, since we have defined all the $2$- and $3$-outer-types in $\str{B}$ by copying (pulling) them from $\str{A}$.
  More precisely, let $\forall x \forall y \forall z \phi_i(x,y,z)$ be one of the $\mathcal{I}_2$-conjuncts (for $\mathcal{I}_1$-conjuncts the situation is analogous).
  Thanks to the uniformity of our logic, all the atoms in $\phi_i$ either use all the variables $x,y,z$ or just one of them.
  Hence, if we evaluate $\phi_i(a,b,c)$ on any three (not necessary distinct) elements $a,b,c \in B$, the value of $\phi_i(a,b,c)$ depends only on: the 1-types of $a,b,c$, and the connections joining all the three elements $a,b,c$ simultaneously.
  Since, the 1-types are fixed beforehand, and the connections joining all the elements $a,b,c$ are pulled as an outer-structure from a valid model $\str{A}$, then $\str{B} \models \phi_i(a,b,c)$ must hold by Lemma \ref{l:preserving}.
  One can also notice that no conflicts were introduced when copying different outer-structures.
  Indeed, in the algorithm we defined outer-structure on every at most three-element subset of $B$ exactly once, and the outer-structures cannot speak about any proper subset of elements of its domain, except the 1-types which where fixed already in \emph{Step 1}.

Let us now explain that the $\mathcal{I}_3$-conjuncts are satisfied.
  Let $\forall x \exists y \forall z \phi_j(x,y,z)$ be any $\mathcal{I}_3$-conjunct. Fix an element an assume that $b \in B_m$.
  We need to see that in the next layer, i.e. $B_{(m+1)\mod 4}$, there exists an element that can serve as a witness for $b$. It is equivalent to having an outgoing edge from $b$ coloured by $j$.
  Indeed, if there exists  such an edge, then the witnessing graph forces us to pick an appropriate pattern from a model $\str{A}$.
  Recall that each layer $B_m$ consists of pairs $(\tau,m)$, where $\tau$ ranges over all generalized types realized in $\str{A}$.
  Therefore, there must be a representative $a$ realizing $\tau$ in $\str{A}$.
  Hence, we have an element $\ff_j(a)$, such that $\str{A} \models \forall z \phi_j(a,\ff_j(a),z)$, which itself realizes some other generalized type $\tau'$.
  Thus, in $B_{(m+1)\mod 4}$ we have an element $b' = (\tau',(m+1)\mod 4)$ satisfying $b \leadsto_j b'$.
  The fact that $\str{B} \models \forall z \phi_j(b,b',z)$ comes from the same line of arguments as in the previous paragraph.
\end{proof}
\noindent %FIXED indent
In the resulting structure $\str{B}$, similarly as for the structure $\str{A}$, we define witness functions $\ff_j: B \to B$ for each $\mathcal{I}_3$-conjunct.
More precisely, we can set $\ff_j(b) = b'$ if $b \leadsto_j b'$; from the construction it follows that $\str{B} \models \forall z \phi_j(b,\ff_i(b),z)$.
Note that the new functions $\ff_j$ are not necessarily injective, but they are without fixpoints and have disjoint images.

Finally, we give a bound on the size of the produced structure $\str{B}$.
The number of generalized types can be trivially bounded by $t(t+1)(|\mathcal{I}_3|+1)t^{|{\mathcal{I}_3}|}$, where $t=|\AAA|$. As mentioned in the Preliminaries (Lemma~\ref{l:ntyp}) $t$ can be bounded by a function exponential in $|\phi|$, and $|\mathcal{I}_3|$ is clearly  polynomial in $|\phi|$. Hence:
\begin{clm}
The domain of $\str{B}$ is bounded exponentially in $|\phi|$.
\end{clm} 

\subsection{Finite model property for AUF$\mathbf{_1^3}$} \label{s:phi2}

Now we build a model $\str{C} \models \phi_1 \wedge \phi_2$. It will be composed out of some number of copies
of the structure $\str{B} \models \phi_1$ from the previous subsection, carefully connected to provide witnesses
for conjuncts from $\phi_2$.  

\smallskip\noindent
\emph{Stage 0: The domain.} Denote $\mathcal{I}=\mathcal{I}_4 \cup \mathcal{I}_5 \cup \mathcal{I}_6$.
Let us form $\str{C}$ over the following finite domain.
$$C=B \times \mathcal{I} \times \{0, 1\} \times \{0, \ldots, 5\}.$$

\noindent

For every $i \in \mathcal{I} , u=0,1$, and $0 \le m \le 5$ make the substructure $\str{C} \restr (B \times \{i\} \times \{u\} \times \{m\})$ isomorphic to $\str{B}$, via
the natural isomorphism sending each $(b,i,u,m)$ to $b$. We also naturally transfer the functions $\ff_j$ and $\pat$ from $\str{B}$ to $\str{C}$, setting $\ff_j((b,i,u,m))=(\ff_j(b), i,u,m)$ and $\pat((b,i,u,m))=\pat(b)$, for
all $b, i, u, m$. For convenience, we split $C$ into five subsets $C_m=B \times \mathcal{I} \times \{0, 1\} \times \{m\}$,
for $m=0,\ldots, 5$.

In the next stages we will successively take care of the all types of conjuncts of $\phi$.
We will describe a strategy for providing witnesses for $\mathcal{I}_4$-, $\mathcal{I}_5-$, and $\mathcal{I}_6$-conjuncts.
Our strategy will guarantee no conflicts, e.g., if a $3$-outer-type is defined
on  a tuple $c_1, c_2, c_3$ to make $c_3$ a witness for $c_1$, $c_2$ for an $\mathcal{I}_5$-conjunct,
then this tuple will not be used any more for a similar task (e.g., to make $c_1$ a witness for $c_2,c_3$).
To design such a strategy we fix an auxiliary injective function $\chc$ taking a subset of $\{0, \ldots, 5\}$ of size two, and extending it by $1$ element from
$\{0, \ldots, 5\}$. Such a function exists by Lemma \ref{l:comb}.
Having provided the witnesses we will carefully complete the structure ensuring that for every $c \in C$, $\ff_j(c)$ remains a witness
for $c$ and the $\mathcal{I}_3$-conjunct indexed by $j$.

\smallskip\noindent
\emph{Stage 1: Witnesses for $\mathcal{I}_4$-conjuncts.} For every  $c \in C$ and every $i \in \mathcal{I}_4$ repeat the following. Let $m$ be the index such
that $c \in C_m$. Let $a=\pat(c)$. Let $a' \in A$ be such that $\str{A} \models \phi_i(a,a')$. It may happen that $a'=a$. Take as $c'$ any element from $B \times \{i \} \times \{0 \} \times \{m+1 (\mod 6) \}$
with the same $1$-type as $a'$. Set $(\str{C} \restr \{c,c'\}) \leftarrow pull(\str{A}, c \leftarrow a, c' \leftarrow a')$.

\smallskip\noindent
\emph{Stage 2: Witnesses for $\mathcal{I}_5$-conjuncts.} For every unordered pair of elements $c_1, c_2 \in C$, $c_1 \not= c_2$ and every $i \in \mathcal{I}_5$ repeat the following. Let $m_1, m_2$ be the indices such that $c_1 \in C_{m_1}$
and $c_2 \in C_{m_2}$.  We
choose two pattern elements $a_1, a_2 \in A$ with the same $1$-types as $c_1$ and, resp., $c_2$, and an index $m'$ as follows.
\begin{enumerate}[label=(\alph*),align=left] %FIXED adjusted list type
	\item If $c_2=\ff_j(c_1)$ for some $j$ (in this case $m_1=m_2$), then let $a_1=\pat(c_1)$, and let $a_2=\ff_j(a_1)$;  set $m'=m_1+1 (\mod 6)$.
		
	\item If $c_1=\ff_j(c_2)$ for some $j$ ($m_1=m_2$), then let $a_2=\pat(c_2)$, and let $a_1=\ff_j(a_2)$;  set $m'=m_1+1 (\mod 6)$.

	\item If none of the above holds, then let $a_1=\pat(c_1)$, $a_2=\pat(c_2)$
	(it may happen that $a_1=a_2$); if $m_1=m_2$, then set $m'=m_1 + 1 (\mod 6)$, otherwise
	set $m'$ to be the only element of  $\chc(\{m_1,m_2\}) \setminus \{m_1, m_2\}$. 
\end{enumerate}
	\noindent %FIXED indent
	Let $a' \in A$ be such that $\str{A} \models \phi_i(a_1,a_2, a')$, and let $a''$ be such that $\str{A} \models \phi_i(a_2,a_1, a'')$.
	Take as $c'$  any element from $B \times \{i \} \times \{0 \} \times \{m' \}$ with the same $1$-type as $a'$, and as $c''$ any element 
from $B \times \{i \} \times \{1 \} \times \{m' \}$ with the same $1$-type as $a''$. Set $\str{C} \restr_* \{c_1,c_2,c'\} \leftarrow pull^*(\str{A}, c_1 \leftarrow a_1,c_2 \leftarrow a_2, c' \leftarrow a'))$, and
set $\str{C} \restr_* \{c_2,c_1,c''\} \leftarrow pull^*(\str{A}, c_2 \leftarrow a_2, c_1 \leftarrow a_1, c'' \leftarrow a'')$.
	
	Additionally, for any $c \in C$ and any $i \in \mathcal{I}_5$, repeat the following. Let $m$ be the index such that $c \in C_m$. Let $a=\pat(c)$. Let $a' \in A$ be such that $\str{A} \models \phi_i(a,a,a')$. 
Take as $c'$ any element from $B \times \{i \} \times \{0 \} \times \{m+1 (\mod 6) \}$
with the same $1$-type as $a'$. Set $(\str{C} \restr \{c,c'\}) \leftarrow pull(\str{A}, c \leftarrow a, c' \leftarrow a')$.

\smallskip\noindent
\emph{Stage 3: Witnesses for $\mathcal{I}_6$-conjuncts.}
For every $c \in C$, and every $i \in \mathcal{I}_6$, repeat the following. Let $m$ be the index such that $c \in C_m$. Let $a = \pat(c)$. 
Let $a', a'' \in A$ be such that $\str{A} \models \phi_i(a,a',a'')$. Take as $c'$ any element in $C_m$, in a different copy
of $\str{B}$ than $c$, with the same $1$-type as $a'$. Take as $c''$ any element in $B \times \{i \} \times \{0 \} \times \{m+1 (\mod 6) \}$ with the same $1$-type as
$a''$, and set $\str{C} \restr_* \{c,c',c''\} \leftarrow pull^*(\str{A}, c \leftarrow a, c' \leftarrow a',c'' \leftarrow a'')$.

\smallskip\noindent
\emph{Stage 4: Taking care of the $\mathcal{I}_3$-conjuncts in $\str{C}$.} For every $c \in C$, for every $c' \in C$ and every $j \in \mathcal{I}_3$ such that
the outer-substructure on $\{c, \ff_j(c), c' \}$ has not yet been defined (note that the elements $c, \ff_j(c), c'$ must be pairwise distinct in this case), find any $a' \in A$ with the same $1$-type as $c'$, take $a=\pat(c)$,
and set $\str{C}\restr_* \{c, \ff_j(c), c' \} \leftarrow pull^*(\str{A}, c \leftarrow a, \ff_j(c) \leftarrow \ff_j(a), c' \leftarrow a')$.

\smallskip\noindent
\emph{Stage 5: Completing $\str{C}$.} For any tuple of pairwise  distinct elements $c_1,c_2,c_3 \in C$ for which 
the outer-substructure of $\str{C}$ on  has not yet been defined, set 
$\str{C}\restr_* \{c_1, c_2, c_3 \} \leftarrow pull^*(\str{A}, c_1 \leftarrow \pat(c_1), c_2 \leftarrow \pat(c_2), c_3 \leftarrow \pat(c_3))$.  
Similarly, for any pair of distinct elements $c_1,c_2 \in C$ for which 
the substructure of $\str{C}$  has not yet been defined, set 
$\str{C}\restr \{c_1, c_2\} \leftarrow pull(\str{A}, c_1 \leftarrow \pat(c_1), c_2 \leftarrow \pat(c_2))$.  

\begin{clm}
$\str{C} \models \phi$.
\end{clm}
\begin{proof}
The purely universal conjuncts of $\phi$ ($\mathcal{I}_1-$ and $\mathcal{I}_2$-conjuncts) are satisfied since every $1$-type, $2$-outer-type
and $3$-outer-type defined in $\str{C}$ is copied (pulled) from $\str{A}$ (either in this subsection or in the previous one where $\str{B}$ was
constructed), which is a model of $\phi$.  All elements have witnesses for $\mathcal{I}_4$- $\mathcal{I}_5$-, and
$\mathcal{I}_6$-conjuncts, as this is guaranteed in Stages 1, 2, and 3, respectively. Finally, for every $c$ and every $j \in \mathcal{I}_3$,
$\ff_j(c)$ is a witness for $c$ and $\forall x \exists y \forall z \phi_j(x,y,z)$ in its copy of $\str{B}$. In Stage 4 we carefully define
$3$-outer-types on tuples consisting of $c, \ff_j(c)$ and elements from the other copies of $\str{B}$, so that $\ff_j(c)$ remains a witness
for $c$ in the whole structure $\str{C}$. 

The only point which remains to be explained is that there are no conflicts in assigning types to tuples of elements. First note
that there are no conflicts when assigning $2$-outer-types. They are defined in Stage 1 and in the last paragraph of Stage 2 (and completed
in Stage 5--but this Stage may not be a source of conflicts as it sets only the types which have not been set before).
The strategy is that the elements from $C_m$ look for witnesses in $C_{m+1 \text{ mod } 5}$,  and if an element $c$ looks for witnesses for 
conjuncts indexed by different $i_1, i_2$, then it always looks for them in different copies of $\str{B}$. (A reader familiar
with \cite{GKV97} may note that the way we deal with $2$-outer-types is similar to the one used in that paper, where three sets
of elements are used and the witnessing scheme is modulo $3$.) Considering $3$-outer-types, they are set in Stages 2, and 3 (plus
the completion in Stage 4 and Stage 5). No outer-type definition from Stage 2 can conflict with an outer-type definition from Stage 3: in Stage 2 we define
outer-types on tuples consisting of elements either belonging two three different $C_m$, or two of them
belongs to the same $C_m$ and the third one to $C_{m+1}$, to a copy of $\str{B}$ indexed by $\mathcal{I}_5$;
in Stage 3 we set outer-types on tuples consisting of two elements from the same $C_m$ and an element from $C_{m+1}$, but from
a copy of $\str{B}$ indexed by $\mathcal{I}_6$. Definitions from the same stage, involving two different $C_m$ do not conflict
with each other due to our circular witnessing scheme requiring elements from $C_m$ to look for witnesses in $C_{m+1}$.
The most interesting case is setting $3$-outer-types for elements from  $3$ different $C_m$-s (Stage 2). This is done without
conflict since the function $\chc$ is injective.
\end{proof}
\noindent %FIXED indent
Since $|C|=8|B||\mathcal{I}|$, $|\mathcal{I}|$ is linear in $|\phi|$ and $|B|$ is exponential in $\phi$
we see that $|C|$ is bounded exponentially in $|\phi|$, which completes the proof of Theorem. \ref{t:fmpauf}.

\section{An excursion into three-variable logic}
\label{s:fothreem}

The main technical advance in this Section will be
showing that the decidability is retained if we relax the uniformity conditions in normal form for  \AUFthree{}, by allowing the formulas $\phi_i$ in $\mathcal{I}_5$-
and $\mathcal{I}_6$-conjuncts to be arbitrary (not necessarily uniform) Boolean combinations of atoms.
We believe that this result is interesting in itself, but additionally it allows us to define a pretty rich decidable class 
of \FOthree{} formulas with nested quantification, which we call \FOthreem{}. Recall that \FOthreem{} is defined in  Preliminaries, where we also observed that it contains both \FOt{} and \AUFthree{}.

Let us see that we are indeed able to reduce \FOthreem{} sentences to normal form resembling that of \AUFthree{}. We say that an \FOthreem{} sentence is in normal form if $\phi=\phi_1 \wedge \phi_2$ for
\begin{eqnarray*}
&\phi_1=  \bigwedge_{i \in \mathcal{I}_2} \forall x \forall y \forall z \phi_i(x,y,z) \wedge 
\bigwedge_{i \in \mathcal{I}_3} \forall x \exists y \forall z \phi_i(x,y,z), \text{ and }\\
&\phi_2= \bigwedge_{i \in \mathcal{I}_5} \forall x \forall y \exists z (x\not=y \rightarrow \phi_i(x,y,z)),
\end{eqnarray*}
where the $\phi_i$ for $i \in \mathcal{I}_2 \cup \mathcal{I}_3$ are positive Boolean combinations of literals using
either one variable or all variables from $\{x,y,z\}$ (uniform combinations), 
and the $\phi_i$ for $i \in \mathcal{I}_5$ are  positive Boolean combinations of arbitrary literals (without equality). 
We note that even though
we generally do not admit equality we allow ourselves to use $x \not= y$ as a premise in $\mathcal{I}_5$-conjuncts;
as we will see this use of equality simplifies model constructions. As those conjuncts arise as equivalent replacements of some equality-free conjuncts $\forall x \forall y \exists z \phi_i'(x,y,z)$ (see the last step of the proof sketch below), we may safely use Lemma \ref{l:preserving}, concerning the pull operation for normal form \FOthreem{} formulas.

\begin{lem}\label{l:nf3}
(i)
The satisfiability problem for \FOthreem{} can be reduced in nondeterministic polynomial time to the satisfiability problem
for normal form \FOthreem{} sentences.
(ii) If the class of all normal form \FOthreem{} sentences has the finite (exponential) model property, then also the whole \FOthreem{} has the
finite (exponential) model property. 
\end{lem}
\begin{proof} (Sketch)
We proceed as in the case of \AUFm{}, introducing new relation symbols to represent subformulas
starting with a block of quantifiers. This time we need symbols of arity $0$, $1$ and $2$. In particular,  subformulas of the form $\exists v \forall v' \eta$ ($\forall v \forall v' \eta$)
with a free variable $v''$ are replaced by $P(v'')$ for a fresh symbol $P$, and axiomatised by 
$\forall v'' \exists v \forall v' (P(v'') \rightarrow \eta)$ ($\forall v'' \forall v \forall v' (P(v'') \rightarrow \eta)$),
and subformulas $\exists v \eta(v, v', v'')$ are replaced by $R(v',v'')$ which is axiomatised by 
$\forall v \forall v' \exists v'' (R(v', v'') \rightarrow \eta))$.
This way, possibly by renaming the variables, we get a formula  of the shape as in \AUFthree{}-normal form but allowing the formulas $\phi_i$ for $i \in \mathcal{I}_5 \cup \mathcal{I}_6$
to be arbitrary (not necessarily uniform) Boolean combinations of atoms. We can now eliminate 
$\mathcal{I}_1$-, $\mathcal{I}_4$- and $\mathcal{I}_6$-conjuncts. 
We first proceed with $\mathcal{I}_6$-conjuncts $\forall x \exists y \exists z \phi_i(x,y,z)$, replacing them
with the conjunctions of $\forall x \exists y R(x,y)$ and $\forall x \forall y \exists z (R(x,y) \rightarrow \phi_i(x,y,z))$, for a 
fresh binary symbol $R$.
Then for  $\mathcal{I}_1$-conjuncts
$\forall x \forall y \phi_i(x,y)$, we just equip them  with a dummy existential quantifier: $\forall x \forall y \exists z \phi_i(x,y)$,
and for $\mathcal{I}_4$-conjuncts $\forall x \exists y \phi_i(x,y)$ we add a dummy universal quantifier and rename variables:
$\forall x \forall y \exists z \phi_i(x,z)$. We end up with only $\mathcal{I}_5$-conjuncts in the $\phi_2$ part of $\phi$.
Finally, any conjunct of the form $\forall x\forall y \exists z \phi_i(x,y,z)$ is replaced by $\forall x \forall y \exists z
(x \not= y \rightarrow \phi_i(x,y,z) \wedge \phi_i(x,x,z))$.
It should be clear that our transformations are sound over the domains consisting of at least two elements. 
(Singleton models can be enumerated and checked separately.)
\end{proof}

\subsection{The complexity of FO$\mathbf{^3_-}$} \label{ss:complexity}

In this section we show that \FOthreem{} is decidable, moreover its satisfiability problem is still in \NExpTime. 
Concerning the obtained results, this subsection is a bit redundant, as in the next subsection, applying an alternative,
probabilistic approach, we demonstrate the exponential model property for \FOthreem{}, which immediately gives the same
\NExpTime-upper bound on the complexity of the satisfiability problem, and additionally proves the same bound on 
the complexity of the finite satisfiability problem.

We decided to keep this subsection (although presenting some constructions and proofs in a slightly more sketchy way), 
since it gives some additional insight in the nature of the considered problems. We proceed here in a rather classical way:
we show that a normal form \FOthreem{} sentence $\phi$ is satisfiable iff it has a certain \emph{witness for satisfiability}, of exponentially bounded size, as stated in the following lemma. Such a witness of satisfiability is a structure which is 
``almost'' a model of $\phi$. It will turn out that one can quite easily (i) extract a witness from a given model of $\phi$, and (ii) unravel a witness
into a proper infinite model of $\phi$. In both directions we use a chase-like procedure. 

\begin{lem} \label{l:satwitness}
A normal form \FOthreem{} sentence $\phi=\phi_1 \wedge \phi_2$ is satisfiable iff there exists a structure $\str{D}$ of size bounded by
$h(|\phi|)$ for some fixed exponential function $h$, and functions without fixed points $\ff_j:D \rightarrow D$ for $j \in \mathcal{I}_3$, such that: 
\begin{enumerate}[label=\textit{(\alph*)},align=left]%FIXED changed numeration to make ensure uniform left alignment
	\item $\str{D} \models \phi_1$
	\item for every $d, d' \in D$ and every $j \in \mathcal{I}_3$, we have  $\str{D} \models \phi_j(d, \ff_j(d), d')$.

	\item for every $i \in \mathcal{I}_5$, for every $2$-type $\beta$ realized in $\str{D}$, there is a realization $(d, d')$ of $\beta$ and an element $d''$ such that $d'' \not=d$, $d'' \not= d'$ and 	$\str{D} \models \phi_i(d,d',d'')$.

	\item[\textit{(c')}] for every $i \in \mathcal{I}_5$, for every $2$-type $\beta$ realized in $\str{D}$ by a pair $(d_0, \ff_j(d_0))$ for some $d_0$ and $j$, there is $d \in D$
	such that for $d'=\ff_j(d)$ the pair $(d,d')$ has type $\beta$, and there is an element $d''$ such that $d'' \not=d$, $d'' \not= d'$ and	$\str{D} \models \phi_i(d,d',d'')$.

\end{enumerate}
\end{lem}

\noindent %FIXED indent
A structure $\str{D}$, together with the functions $\ff_j$ meeting the conditions of the above lemma will be called \emph{a witness of satisfiability for $\phi$}. As we see $\str{D}$ satisfies the $\mathcal{I}_2$- and $\mathcal{I}_3$-conjuncts of $\phi$ (Conditions (a), (b)), but
need not to satisfy the $\mathcal{I}_5$-conjuncts. However, if a pair $d_0, d_0'$ does not have a witness for an
$\mathcal{I}_5$-conjunct, then
there is an appropriate ``similar'' pair $d, d'$ having such a witness (Conditions (c), (c')). 

Let us first prove the left-to-right direction of Lemma \ref{l:satwitness}.

\smallskip\noindent
{\bf Left-to-right: extracting a witness of satisfiability from a model.}

\smallskip\noindent
\emph{Stage 1: The building block $\str{B}$.} Let $\str{A}$ be a model of a normal form $\FOthreem$ formula $\phi = \phi_1 \wedge \phi_2$.
Let $\AAA$ and $\BBB$ be the sets of all 1-types and all 2-outer-types realized in $\str{A}$, respectively.
Since our logic does not support equality, w.l.o.g.~we assume that for each 1-type $\alpha \in \AAA$, there exists a pair $a,b\in\str{A}$ realizing the 2-outer-type $\beta \in \BBB$ which ``unravel'' $\alpha$,
which means that: $\type{\str{A}}{a} = \type{\str{A}}{b} = \alpha$, and for each relation symbol $R$ in the signature it holds: $\str{A} \models R(a,b)$ iff $\str{A} \models R(a,a)$, and $\str{A} \models R(b,a)$ iff $\str{A} \models R(a,a)$.

Let $\str{B}$ be the finite structure constructed for $\phi_1$ and $\str{A}$ as in Section~\ref{s:phi1}. 
It will serve as a building block in the construction of a witness of satisfiability $\str{D}$ for $\phi$.
Inspecting the construction of $\str{B}$ we note that the set of $1$-types realized in $\str{B}$ is precisely
$\AAA$, and the set of $2$-types is contained in $\BBB$; and regarding the size of $\str{B}$, it is 
exponentially bounded in $|\phi|$. 
Let  $\ff_j:B \to B$ for $j \in \mathcal{I}_3$ be the witness functions defined as in Section~\ref{s:phi1}.

\smallskip\noindent
\emph{Stage 2: The domain of $\str{D}_s$.}
The witness of satisfiability $\str{D}$ will consists of some number of copies of the structure $\str{B}$, say $k$ (will be clarified later).
More precisely, we will construct a sequence of structures $\str{D}_0, \str{D}_1, \dots, \str{D}_{k-1}$, where $\str{D}_0$ is isomorphic to $\str{B}$ and $\str{D}_{k-1} = \str{D}$.
Each $\str{D}_{s+1}$ will be created from $\str{D}_s$ by attaching a fresh disjoint copy of $\str{B}$, providing an $\mathcal{I}_5$-witness for some pair of elements of $(\str{D}_{s+1} \restr D_s)$, and then completing the structure.

Formally, the domain of $\str{D}_s$ is the set $$D_s=B \times \{0, \dots, s \}.$$
We make the structure $\str{D}_s \restr (B \times \{s\})$ isomorphic to $\str{B}$,
and for $s>0$ we make the structure $\str{D}_s \restr (B \times \{0,\dots,s-1\})$ isomorphic to $\str{D}_{s-1}$.
We should think that $B \times \{0\}$ is the ``main'' copy of $\str{B}$, while $B \times \{m\}$, for $m>0$, are auxiliary copies that will help us to satisfy the conditions (c) and (c').
We can adopt the functions $\ff_j:B \to B$ to the functions $\ff_j:D_s \to D_s$ by letting $\ff_j((b,m)) = (\ff_j(b),m)$.

We introduce a notion of an \emph{extended 2-type} $(\beta, j)  \in \BBB \times (\mathcal{I}_3 \cup \{ \bot \})$, which in addition to the information about the 2-outer-type contains an information about whether one of the element (of a pair realizing this $2$-type) is an $\mathcal{I}_3$-witness for the other. 
We say that $a_1,a_2 \in \str{A}$ ($a_1, a_2 \in \str{D}_s$) realize an extended 2-type $(\beta, j)$ in $\str{A}$ (resp.~$\str{D}_s$) if
(1) $a_1,a_2$ realize 2-outer-type $\beta$, and (2) either $i = \bot$, or $\ff_j(a_1) = a_2$, or $\ff_j(a_2) = a_1$.
Denote by $\BBB^+$ the set of all  extended 2-types realized in $\str{A}$.
For each $i\in \mathcal{I}_5$ select a pattern function $\pat^{\str{A}}_i: \BBB^+ \to A^3$ such that $\pat^{\str{A}}_{i}(\beta, j) = (a_1, a_2, a_3)$ if $\{a_1, a_2\}$ realizes $(\beta, j)$ in $\str{A}$, and $\str{A} \models \phi_i(a_1, a_2, a_3).$

\smallskip\noindent
\emph{Stage 3: Constructing $\str{D}_{s+1}$ from $\str{D}_s$.}
Let $b_1, b_2 \in \str{D}_s$ be a pair of distinct elements realizing an extended 2-type $(\beta, j)$ for which  (c) or (c') fails
for some $i \in \mathcal{I}_5$.
Let $(a_1, a_2, a_3) = \pat^{\str{A}}_i(\beta, j)$. Set $(\str{D}_{s+1} \restr \{ b_1, b_2, b_3 \}) \leftarrow pull(\str{A}, b_1 \leftarrow a_1, b_2 \leftarrow a_2, b_3 \leftarrow a_3),$
where $b_3 \in B \times \{s+1\}$ is any element of the freshly added copy of $\str{B}$ with the same 1-type as $a_3$.
After providing a witness for $b_1, b_2$ this way, we need to complete $\str{D}$, i.e., we need to specify the relations between all pairs and triples of elements for which we have not done it yet.
Take any subset $S \subseteq D_{s+1}$ of size $2$ or $3$, and search for an appropriate pattern in $\str{A}$, then pull it into $\str{D}_{s+1} \restr S$.
We will omit the details here. Consult the Section~\ref{s:phi1} and the algorithm defined there for the very similar line of the reasoning.
Therefore, $\str{D}_{s+1} \models \phi_1$ (Condition (a)), and $\str{D}_{s+1} \models \phi_i(d, \ff_j(d), d')$ for each $j\in \mathcal{I}_3$ (Condition (b)).

\smallskip\noindent
\emph{Stage 4: Finalization.}
Conditions (a) and (b) are satisfied after every step, and after at most $|\BBB^+|$ many steps (the number $k$), which is polynomial
in $|\BBB|$ and $|\phi|$,  all the $2$-types will satisfy (c) and (c'). Thus, all the conditions are met for $\str{D}_k$. 
The size of the domain of $\str{D}_k$ is at most exponential in $|\phi|$, since $\str{D}_k$ consists in fact of $|\BBB^+|$ copies
of $\str{B}$, and both $\str{B}$ and $|\BBB^+|$ are exponentially bounded in $\phi$.

\smallskip
Let us now turn to the right-to-left direction of Lemma \ref{l:satwitness}.

\smallskip\noindent
{\bf Right-to-left: building a model from a witness of satisfiability}

\smallskip
Let $\str{D}$ be a witnesses of satisfiability.
If $\phi$ does not have any $\mathcal{I}_5$-conjuncts, then $\str{D}$ is already a model of $\phi$.
In the other case, we will show how to construct a new witness of satisfiability $\str{D}'$, but bigger.
Iterating this construction indefinitely, and taking the limit structure, we will obtain an infinite model of $\str{A}$.

The construction of a new witnesses of satisfiability is very similar to the process of extracting a witnesses of satisfiability from a model. Hence, we will just sketch the construction, and stress the differences.
The domain of $\str{D}'$ is the set\footnote{By $\binom{A}{2}$, where $A$ is a set, we mean the set of all $2$-element subsets of $A$. }
$$D'=D \times \{ (\bot, \emptyset) \} \cup D \times \mathcal{I}_5 \times {\binom{D}{2}},$$ 
and let the structure on $\str{D}' \restr D \times \{s\}$ be isomorphic to $\str{D}$ for each $s \in \{(\bot, \emptyset)\} \cup (\mathcal{I}_5 \times \binom{D}{2})$.
Then we take care of witnesses for all $\mathcal{I}_5$-conjuncts for all the pairs of elements in $\str{D}_{\bot, \emptyset}$. The witness for a conjunct $\phi_i$ and a pair $\{d_1, d_2\}$ is provided in $D_{i, \{d_1, d_2\}}$.
The major difference is that this time we pull outer-substructures not from an existing model of $\phi$, but rather from $\str{D}$ itself;
this is always possible due to Conditions (c) and (c').
Then we complete the structure, similarly as in our previous constructions, by setting the relations between all the pairs and the triples of elements for which we have not done it yet.
Again, in this completion process the outer-substructures are pulled from $\str{D}$, and not from an external structure.
So after that all the $\mathcal{I}_2$- and $\mathcal{I}_3$-conjuncts are satisfied, and all the pairs of elements from $D$
have witnesses for all $\mathcal{I}_5$-conjuncts.
We claim that the newly obtained structure is again a witness of satisfiability, with all elements from 
$D$ having the required witnesses, but with a strictly larger domain.
This allows us to continue the process ad infinitum, with the natural limit structure being a model of $\phi$.

\smallskip
Having shown Lemma \ref{l:satwitness}, and noting that checking if a given structure $\str{D}$ is a witness of satisfiability 
for $\phi$ can be easily done in time polynomial in $|D|$ and $|\phi|$, we get:
\begin{thm}
The satisfiability problem for \FOthreem{} is \NExpTime-complete.
\end{thm}
\noindent %FIXED indent
Note that in this subsection we say nothing about finite satisfiability for \FOthreem{} or about its finite model property.
These are left for the next subsection.

\subsection{Finite model property via probabilistic method} \label{ss:fmp}
Here we show show that \FOthreem{} has the finite model property. As in Section~\ref{s:probabilistic} we employ the probabilistic method.
This time we will use as finite building blocks the structures $\str{B} \models \phi_1$ constructed deterministically. Randomness will come to play
when those blocks are joined.

Let  $\phi=\phi_1 \wedge \phi_2$ be a normal form \FOthreem{}-sentence and  $\str{A}$ its model. 
W.l.o.g.~we assume that any $1$-type realized in $\str{A}$ is realized by at least three distinct elements (if it is not the case, 
we replace $\str{A}$ by the structure $\str{A}'$ with domain $A \times \N$, as in the beginning of Section~\ref{s:probabilistic}).
As previously, we pick, for every $j \in \mathcal{I}_3$, a witness function $\ff_j:A \rightarrow A$
such that, for every $a \in A$, we have $\str{A} \models \forall z \phi_j(a, \ff_j(a), z)$.
Let $\AAA{}, \BBB{}$ and $\CCC{}^*$ be the set of $1$-types realized in $\str{A}$, the set of $2$-types realized in $\str{A}$, and, resp., the
set of $3$-outer-types realized in $\str{A}$. 

Let us construct a model $\str{B} \models \phi_1$, together with functions $\pat:B \rightarrow A$ and $\ff_j:B \rightarrow B$, for $j=1, \ldots, |\mathcal{I}_3|$, the latter without fixed points, as in Section~\ref{s:phi1}. We may assume that the set of
$1$-types realized in $\str{B}$ is $\AAA$ (cf.~Section~\ref{ss:complexity}, Left-to-right, Stage 1).

We now define a sequence of (partly) random structures $\str{C}_n$, for $n=1,2, \ldots$. 

\smallskip\noindent
{\em Stage 0: The domain.}  $C_n=B \times \{1, \ldots, n\}$. We make the substructures $\str{C}_n \restr B \times \{m\}$ 
isomorphic to $\str{B}$, for all $m=1,2, \ldots, n$, and set no other connections in this stage. We naturally lift 
functions $\pat$ and the $\ff_j$ to $\str{C}_n$.

\smallskip\noindent
{\em Stage 1: $2$-types.} For every pair $c_1, c_2 \in C_n$ of distinct elements for which  $\type{\str{C}_n}{c_1,c_2}$ has not yet been defined (note that $c_1$ and $c_2$ belong then to different copies of $\str{B}$ in this case), proceed as follows.
Let $\alpha_1=\type{\str{C}_n}{c_1}$, $\alpha_2=\type{\str{C}_n}{c_2}$, and let $\BBB_0$ be the set of $2$-types $\beta$ realized in $\str{A}$ such that
$\beta \restr x_1=\alpha_1$ and $\beta \restr x_2 =\alpha_2$. 
Clearly $\BBB_0$ is non-empty, since even if $\alpha_1=\alpha_2$ we will have two distinct realizations of this
$1$-type by our assumption about $\str{A}$.
Choose $\type{\str{C}_n}{c_1,c_2}$ at random from $\BBB_0$, with uniform probability. 

\smallskip\noindent
{\em Stage 2a: $3$-outer-types.} For every tuple $c_1, c_2, c_3 \in C_n$  such that $c_2=\ff_j(c_1)$ for some $j$, and $\outertype{\str{C}_n}{c_1,c_2,c_3}$ has 
not yet been defined, proceed as follows. 
Recall that by the construction of $\str{B}$, $c_2$ may not
be in the image of any $\ff_k$ for $j\not=k$ in this case.
Let $a_1=\pat(c_1)$, $a_2=\ff_j(a_1)$, and $\alpha_3=\type{\str{C}_n}{c_3}$. Let $\CCC^*_0$ be the set of $3$-outer-types realized in $\str{A}$ by 
tuples $a_1, a_2, a_3$, for all $a_3 \in A$ having $1$-type $\alpha_3$. Note that $\CCC^*_0$ is non-empty--even if $a_1$ and $a_2$ have type $\alpha_3$, then at least one appropriate
$a_3$, different from $a_1$, $a_2$ exists since we assumed that every $1$-type realized in $\str{A}$ is realized at least three times.
Choose $\outertype{\str{C}_n}{c_1,c_2,c_3}$  at random from $\CCC^*_0$, with uniform probability. 

\smallskip\noindent
{\em Stage 2b: $3$-outer-types contd}. For every tuple $c_1, c_2, c_3 \in C_n$ for which $\outertype{\str{C}_n}{c_1,c_2,c_3}$ has not yet been defined, proceed as follows.
Let $\alpha_1=\type{\str{C_n}}{a_1}$, $\alpha_2=\type{\str{C_n}}{a_2}, \alpha_3=\type{\str{C_n}}{a_3}$,  and let $\CCC^*_0$ be the set of $3$-outer-types $\gamma^*$ realized in $\str{A}$ such that
$\gamma^* \restr x_1=\alpha_1$, $\gamma^* \restr x_2 =\alpha_2$, and $\gamma^* \restr x_3 =\alpha_3$. Again, using the assumption
about $\str{A}$ we note that $\CCC^*_0$ is non-empty. Choose $\outertype{\str{C}_n}{c_1,c_2,c_3}$ at random from $\CCC^*_0$, with uniform probability.

\smallskip

Let us now estimate the probability that $\str{C}_n \not\models \phi$. It is not difficult to observe that, regardless of our random choices, $\str{C}_n \models \phi_1$. Indeed, the conjuncts indexed by $\mathcal{I}_3$ are respected since they are respected
in $\str{B}$, and we ensure that $\ff_j(c)$ remains an appropriate witness for $c$ in $\str{C}_n$, for every $c$ in Stage 2a; the conjuncts indexed by $\mathcal{I}_2$
are respected since they are respected in $\str{B}$, and all outer-types defined in $\str{C}$ are copied from $\str{A}$ (Stages 2a, 2b). 
Hence, we need to take into consideration only the satisfaction of $\phi_2$. Let $F^i_{c_1, c_2}$, for a pair of distinct elements  $c_1, c_2 \in C_n$ and
$i \in \mathcal{I}_5$, denotes the event ``$\str{C}_n \not\models \exists z \phi_i(c_1, c_2, z)$''.   We first estimate its probability.

Take a pair of distinct elements $c_1, c_2 \in C_n$ and $i \in \mathcal{I}_5$. There are three cases: (a) $c_2=f_j(c_1)$ for some $j$, (b) $c_1=f_j(c_2)$ for some $j$,
and (c) none of the above holds. Let us go in details through case (a). Let $a_1=\pat(c_1)$, $a_2=\ff_j(a_1)$. In this case $c_1$ and $c_2$ belong to the same copy of $\str{B}$, say to $\str{B} \times \{k\}$, and $\type{\str{C}_n}{c_1,c_2}$ was set in Stage 0 to be equal to  $\type{\str{A}}{a_1, a_2}$.
If $\str{A} \models \phi_i(a_1, a_2, a_1)$, then $\str{C}_n \models \phi_i(c_1, c_2, c_1)$, and if 
$\str{A} \models \phi_i(a_1, a_2, a_2)$, then $\str{C}_n \models \phi_i(c_1, c_2, c_2)$, so the pair $a_1, a_2$ has a witness for the $\phi_i$. Otherwise,
there is an element $a_3$, $a_3 \not\in \{a_1, a_2\}$ such that $\str{A} \models \phi_i(a_1, a_2, a_3)$. Let $\alpha_3=\type{\str{A}}{a_3}$.
Consider now any element $c_3$ of type $\alpha_3$ in $\str{B} \times (\{1, \ldots k-1, k+1, \ldots n)$. During the construction of $\str{C}_n$ we randomly chose $\type{\str{C}_n}{c_1, c_3}$, $\type{\str{C}_n}{c_2, c_3}$ (Stage 1), $\outertype{\str{C}_n}{c_1, c_2, c_3}$ (Stage 2a). We had at most $|\BBB|$ choices for $\type{\str{C}_n}{c_1, c_3}$  and at least one of  them was $\type{\str{A}}{a_1, a_3}$. Analogously, for $\type{\str{C}_n}{c_2, c_3}$.
We had at most $|\CCC^*|$ choices for $\outertype{\str{C}_n}{c_1, c_2,c_3}$, and at least one of them was $\outertype{\str{A}}{a_1,a_2,a_3}$. Hence
the probability that $\{c_1 \rightarrow a_1, c_2 \rightarrow a_2, c_3 \rightarrow a_3\}$ is a partial isomorphism, and hence $\str{C}_n \models \phi_i(c_1, c_2, c_3)$
is at least $\epsilon=1/(|\BBB|^2 \cdot |\CCC^*|)$. As there are at least $n-1$ elements of type $\alpha_3$ in $\str{B} \times (\{1, \ldots k-1, k+1, \ldots n)$ (at least one in each copy of $\str{B}$), and the structures on $c_1,c_2,c_3$ and $c_1,c_2,c_3'$ for different such
elements $c_3,c_3'$ are chosen independently, it follows that $\Prob[F^i_{c_1, c_2}] \le (1-\epsilon)^{n-1}$. 

Note that the event ``$\str{C}_n \not\models \phi$'' is equal to $\bigcup_{c_1 \not=c_2} \bigcup_{i \in \mathcal{I}_5} F^i_{c_1,c_2}$. Hence $$\Prob [\str{C}_n \not\models \phi] = \Prob[\bigcup_{c_1 \not=c_2} \bigcup_{i \in \mathcal{I}_5} F^i_{c_1,c_2}] \le
\sum_{c_1 \not=c_2}\sum_{i \in \mathcal{I}_5} \Prob[F^i_{c_1,c_2}]=n|B|(n|B|-1)|\mathcal{I}_5|(1-\epsilon)^{n-1}.$$
The limit of the above estimation when $n$ approaches infinity is $0$ (note that $\epsilon$, $|B|$ and $|\mathcal{I}_5|$ do not depend on $n$). In particular, for some $n \in \N$, we have that
$\Prob [\str{C}_n \not\models \phi] < 1$, and thus one of the possible (finite number) of choices of $\str{C}_n$ is indeed a finite model of $\phi$.

We can refine this reasoning and achieve an exponential bound on the size of a minimal finite model by applying the following claim:

\begin{clm}
  There exists an exponential function $f$ such that $\Prob[\str{C}_n \not\models \phi] < 1$ for $n \ge f(|\phi|)$.
\end{clm}
\begin{proof}
  Denote by $|\phi|$ the size of $\phi$ in any reasonable encoding.
  In Section~\ref{s:phi1} we estimated the size of the domain of $\str{B}$ to be exponential in $|\phi|$. Hence, there exist a polynomial $p_1(x)$ such that $|B| \le 2^{p_1(|\phi|)}$.
  Similarly, from Lemma \ref{l:ntyp} we get that there exists polynomials $q_1(x)$ and $q_2(x)$ such that $|\BBB| \le 2^{q_1(|\phi|)}$ and $|\CCC^*| \le 2^{q_2(|\phi|)}$.
  W.l.o.g. we can assume that $p_1(x),q_1(x),q_2(x) \ge 10x$.

  Let us state the bound on the probability of the event ``$\str{C}_{n+1} \not\models \phi$'' in terms of $p(x) := 2p_1(x),$ and $q(x) := 2q_1(x) + q_2(x)$:
  \begin{equation}\label{e:siema}\Prob [\str{C}_{n+1} \not\models \phi] = (n+1)|B|((n+1)|B|-1)|\mathcal{I}_5|(1-\epsilon)^{n} \le 10n^2 2^{p(|\phi|)}|\phi|(1-\varepsilon(|\phi|))^{n}, \end{equation}
    where $\varepsilon(x)$ is defined as $\varepsilon(x) := 1/{2^{q(x)}}$, and therefore $\varepsilon(|\phi|) \le 1/(|\BBB|^2\cdot|\CCC^*|) = \epsilon$.

  We need to show that taking $n=f(|\phi|)$, for some exponential function $f$, makes the right-hand side of (\ref{e:siema}) strictly less than one.\footnote{We did not tried to optimize this bound for the sake of simplicity.}
  Let $f(x) = 10\ln(2)\cdot2^{2q(x)}p(x)$. Then, taking log on both sides of:
  $$ 10f(x)^2 2^{p(x)} x (1 - \varepsilon(x))^{f(x)} < 1 $$
  gives us:
  $$ \log(10) + 2\log(f(x)) + p(x) + \log(x) + f(x)\log(1 - \varepsilon(x)) < 0. $$
  Recall that the Mercator Series states that $\ln(1-x) = -\sum_{i=1}^{\infty}{x^i/i}$ for all $x \in (-1,1)$.
  Hence we want to see that:
  $$ \log(10) + 2\log(f(x)) + p(x) + \log(x) < \frac{f(x)}{\ln(2)}\sum_{i=1}^{\infty}{\frac{\varepsilon(x)^i}{i}}. $$
  Since $f(x)$ and $\varepsilon(x)$ are always positive, the above inequality is implied by:
  $$ \log(10) + 2\log(f(x)) + p(x) + \log(x) < f(x)\varepsilon(x)/\ln(2). $$
  Expanding the definition of $\varepsilon(x)$ and $f(x)$:
  $$ \log(10) + 2\log(10\ln(2)) + 4q(x) + 2\log(p(x)) + p(x) + \log(x) < 10\ln(2)\cdot2^{2q(x)}p(x)\cdot2^{-q(x)}/\ln(2). $$
  Finally we get:
  $$ \log(10) + 2\log(10\ln(2)) + 4q(x) + 2\log(p(x)) + p(x) + \log(x) < 10\cdot2^{q(x)}p(x), $$
  and it is easy to see that each term on the left-hand side is bounded by $2^{q(x)}p(x)$ (possibly for a large enough $x$),
	so the inequality holds.
\end{proof}
\noindent %FIXED indent
Recalling Lemma \ref{l:nf3} we conclude.

\begin{thm}
\FOthreem{} has the finite (exponential) model property. The finite satisfiability problem (=satisfiability problem) for \FOthreem{} is \NExpTime-complete.
\end{thm}

\section{Conclusions}

We proposed an extension of the uniform one-dimensional fragment by mixed blocks of quantifiers, \AUF.
We identified its two subfragments: \AUFm{} and \AUFthree{} with exponential model property
and \NExpTime-complete satisfiability problem. We also extended \AUFthree{} to a more expressive decidable fragment of \FOthree{}. The main open question which is left is whether the whole  \AUF{} is decidable.
We believe that the techniques developed in this paper can be refined to cover the case of the extension of \AUFm{} admitting blocks of quantifiers ending with $\exists \forall$. These extension would actually
capture also \AUFthree{}. The general case seems to require some deeper insight.

Another  question is what happens to our logics if a use of equality (free or uniform) is allowed.
We conclude this paper with some notes concerning this question.

\label{s:conc}
\subsection{Notes on equality}
All the results are obtained in this paper in the absence of equality. 
We have two options of adding equality: first, we can treat the equality symbol
as relation symbols from the signature and allow only its uniform use, as in the formula $\forall xy (P(x) \wedge P(y) \rightarrow x=y)$; second, we can allow for
free use of equality, as e.g. in the formula $\forall xyz (R(x,y,z) \rightarrow x \not=y)$.
We recall that \UF{} with free use of equality is still decidable and has the exponential model property \cite{KK14}.

Regarding the  free use of equality consider the following \AUFthree{} formula:
$$\exists x S(x) \wedge \forall x \exists y \forall z ( \neg S(y) \wedge R(x,y,z)  \wedge (x=z \vee \neg R(z,y,x))$$
It is satisfied in the model whose universe is the set of natural numbers, $S$ is true only at $0$ and $R(x,y,z)$ is 
true iff $y=x+1$. It is readily verified that there are no finite models (every $x$ needs to take a fresh $y$ as a witness since otherwise $R(z,y,x)$ would
 not hold for one of the earlier $z$, $z \not=x$).

The above example can be simply adapted to the case of \AUFm{} with free use of equality. We just add a dummy existentially quantified variable $t$, and
require it to be equal to the previous, universally quantified variable $z$. To accommodate all the variables we increase the arity of $R$ by $2$ (one can think
that the first and the last position of $R$ from the previous example have been doubled):
$$\exists x S(x) \wedge \forall x \exists y \forall z \exists t. (t=z \wedge \neg S(y) \wedge R(x,x,y,z,t)  \wedge (x=z \vee \neg R(z,t,y,x,x))).$$

These examples show that satisfiability and finite satisfiability
are different problems for \AUFm{}/\AUFthree{} with free use of equality. We do not know if they are decidable, however. This issue is left for further investigations.

As for the uniform use, we suspect that it does not spoil the finite model property. Such uniform use 
allow one to say that some $1$-types are realized at most once. We believe that treating the realizations of such types
with appropriate care  (extending the approach with \emph{kings} and \emph{court} for \FOt{} from \cite{GKV97}) one can show that the
finite model property survives.  We leave the details to be
worked out.

It is clear that allowing equality in \FOthreem{} spoils the decidability, as \FOthreem{} contains formulas
of the form $\forall x \forall y \exists z \phi$, for quantifier-free $\phi$, which form an
undecidable subclass of G\"odel class with equality \cite{Gol84}.

\section*{Acknowledgements} This work is supported by  NCN grant No. 2021/41/B/ST6/00996. The authors thank the reviewers of this manuscript and the reviewers of the conference papers \cite{FK23} and \cite{Kie23} (on which this article is based) for their helpful comments.

\bibliographystyle{alphaurl}
\bibliography{mybib}

\newcommand{\etalchar}[1]{$^{#1}$}
\begin{thebibliography}{BDM{\etalchar{+}}11}

\bibitem[AvBN98]{ABN98}
Hajnal Andr{\'e}ka, Johan van Benthem, and Istv\'{a}n N\'{e}meti.
\newblock Modal languages and bounded fragments of predicate logic.
\newblock {\em Journal of Philosophical Logic}, 27:217--274, 1998.
\newblock \href {https://doi.org/10.1023/A:1004275029985}
  {\path{doi:10.1023/A:1004275029985}}.

\bibitem[BBC{\etalchar{+}}16]{BBC16}
Saguy Benaim, Michael Benedikt, Witold Charatonik, Emanuel Kieronski, Rastislav
  Lenhardt, Filip Mazowiecki, and James Worrell.
\newblock Complexity of two-variable logic on finite trees.
\newblock {\em {ACM} Trans. Comput. Log.}, 17(4):32:1--32:38, 2016.
\newblock \href {https://doi.org/10.1145/2996796} {\path{doi:10.1145/2996796}}.

\bibitem[BDM{\etalchar{+}}11]{BDM11}
Miko\l{}aj Boja{\'{n}}czyk, Claire David, Anca Muscholl, Thomas Schwentick, and
  Luc Segoufin.
\newblock Two-variable logic on data words.
\newblock {\em {ACM} Trans. Comput. Log.}, 12(4):27, 2011.
\newblock \href {https://doi.org/10.1145/1970398.1970403}
  {\path{doi:10.1145/1970398.1970403}}.

\bibitem[BKP23]{BKP23}
Bartosz Bednarczyk, Daumantas Kojelis, and Ian Pratt{-}Hartmann.
\newblock On the limits of decision: the adjacent fragment of first-order
  logic.
\newblock In {\em 50th International Colloquium on Automata, Languages, and
  Programming, {ICALP} 2023}, volume 261 of {\em LIPIcs}, pages 111:1--111:21,
  2023.
\newblock \href {https://doi.org/10.4230/LIPICS.ICALP.2023.111}
  {\path{doi:10.4230/LIPICS.ICALP.2023.111}}.

\bibitem[BM17]{BM17}
Simone Bova and Fabio Mogavero.
\newblock Herbrand property, finite quasi-herbrand models, and a chandra-merlin
  theorem for quantified conjunctive queries.
\newblock In {\em 32nd Annual {ACM/IEEE} Symposium on Logic in Computer
  Science, {LICS} 2017}, pages 1--12. {IEEE} Computer Society, 2017.
\newblock \href {https://doi.org/10.1109/LICS.2017.8005073}
  {\path{doi:10.1109/LICS.2017.8005073}}.

\bibitem[BMP17]{BMP17}
Pierre Bourhis, Michael Morak, and Andreas Pieris.
\newblock Making cross products and guarded ontology languages compatible.
\newblock In {\em 26th International Joint Conference on Artificial
  Intelligence, {IJCAI} 2017}, pages 880--886, 2017.
\newblock \href {https://doi.org/10.24963/IJCAI.2017/122}
  {\path{doi:10.24963/IJCAI.2017/122}}.

\bibitem[BMSS09]{BMS09}
Miko\l{}aj Bojanczyk, Anca Muscholl, Thomas Schwentick, and Luc Segoufin.
\newblock Two-variable logic on data trees and {XML} reasoning.
\newblock {\em J. ACM}, 56(3), 2009.
\newblock \href {https://doi.org/10.1145/1516512.1516515}
  {\path{doi:10.1145/1516512.1516515}}.

\bibitem[BtCS15]{BtCS15}
Vince B{\'{a}}r{\'{a}}ny, Balder ten Cate, and Luc Segoufin.
\newblock Guarded negation.
\newblock {\em J. {ACM}}, 62(3):22, 2015.
\newblock \href {https://doi.org/10.1145/2701414} {\path{doi:10.1145/2701414}}.

\bibitem[CW16a]{CW16a}
Witold Charatonik and Piotr Witkowski.
\newblock Two-variable logic with counting and a linear order.
\newblock {\em Log. Methods Comput. Sci.}, 12(2), 2016.
\newblock \href {https://doi.org/10.2168/LMCS-12(2:8)2016}
  {\path{doi:10.2168/LMCS-12(2:8)2016}}.

\bibitem[CW16b]{CW16}
Witold Charatonik and Piotr Witkowski.
\newblock Two-variable logic with counting and trees.
\newblock {\em {ACM} Trans. Comput. Log.}, 17(4):31, 2016.
\newblock \href {https://doi.org/10.1145/2983622} {\path{doi:10.1145/2983622}}.

\bibitem[FK23]{FK23}
Oskar Fiuk and Emanuel Kieronski.
\newblock An excursion to the border of decidability: between two- and
  three-variable logic.
\newblock In {\em {LPAR} 2023: Proceedings of 24th International Conference on
  Logic for Programming, Artificial Intelligence and Reasoning}, volume~94 of
  {\em EPiC Series in Computing}, pages 205--223, 2023.
\newblock \href {https://doi.org/10.29007/1XNS} {\path{doi:10.29007/1XNS}}.

\bibitem[FKM24]{FKM24}
Oskar Fiuk, Emanuel Kieronski, and Vincent Michielini.
\newblock On the complexity of \mbox{M}aslov class \mbox{K}.
\newblock In {\em Proceedings of the 39th Annual {ACM/IEEE} Symposium on Logic
  in Computer Science, {LICS} 2024}, pages 35:1--35:14, 2024.
\newblock \href {https://doi.org/10.1145/3661814.3662097}
  {\path{doi:10.1145/3661814.3662097}}.

\bibitem[GKV97]{GKV97}
Erich Gr{\"a}del, Phokion Kolaitis, and Moshe~Y. Vardi.
\newblock On the decision problem for two-variable first-order logic.
\newblock {\em Bulletin of Symbolic Logic}, 3(1):53--69, 1997.
\newblock \href {https://doi.org/10.2307/421196} {\path{doi:10.2307/421196}}.

\bibitem[Gol84]{Gol84}
Warren~D. Goldfarb.
\newblock The unsolvability of the \mbox{G}\"o\-del class with identity.
\newblock {\em J. Symb. Logic}, 49:1237--1252, 1984.
\newblock \href {https://doi.org/10.2307/2274274} {\path{doi:10.2307/2274274}}.

\bibitem[GOR97]{GOR97}
Erich Gr{\"a}del, Martin Otto, and Eric Rosen.
\newblock Two-variable logic with counting is decidable.
\newblock In {\em 12th Annual {IEEE} Symposium on Logic in Computer Science,
  {LICS} 1997}, pages 306--317. IEEE, 1997.
\newblock \href {https://doi.org/10.1109/LICS.1997.614957}
  {\path{doi:10.1109/LICS.1997.614957}}.

\bibitem[GS83]{GS83}
Yuri Gurevich and Saharon Shelah.
\newblock Random models and the {G}\"odel case of the decision problem.
\newblock {\em J. Symbolic Logic}, 48(4):1120--1124, 1983.
\newblock \href {https://doi.org/10.2307/2273674} {\path{doi:10.2307/2273674}}.

\bibitem[Hal35]{Hal35}
Philip Hall.
\newblock On representatives of subsets.
\newblock {\em Journal of the London Mathematical Society}, s1-10(1):26--30,
  1935.
\newblock \href {https://doi.org/10.1112/jlms/s1-10.37.26}
  {\path{doi:10.1112/jlms/s1-10.37.26}}.

\bibitem[HK14]{HK14}
Lauri Hella and Antti Kuusisto.
\newblock One-dimensional fragment of first-order logic.
\newblock In {\em Proceedings of Advances in Modal Logic, 2014}, pages
  274--293, 2014.

\bibitem[HS99]{HS99}
Ullrich Hustadt and Renate Schmidt.
\newblock Maslov's class {K} revisited.
\newblock In {\em Automated Deduction --- CADE-16}, pages 172--186, 1999.
\newblock \href {https://doi.org/10.1007/3-540-48660-7\_12}
  {\path{doi:10.1007/3-540-48660-7\_12}}.

\bibitem[Kaz06]{Kaz06}
Yevgeny Kazakov.
\newblock {\em Saturation-based decision procedures for extensions of the
  guarded fragment}.
\newblock PhD thesis, Universit\"at des Saarlandes, Saarbr\"ucken, Germany,
  2006.

\bibitem[Kie05]{Kie05}
Emanuel Kiero\'{n}ski.
\newblock Results on the guarded fragment with equivalence or transitive
  relations.
\newblock In {\em 19th International Workshop, {CSL} 2005, 14th Annual
  Conference of the {EACSL}}, volume 3634 of {\em LNCS}, pages 309--324.
  Springer, 2005.
\newblock \href {https://doi.org/10.1007/11538363\_22}
  {\path{doi:10.1007/11538363\_22}}.

\bibitem[Kie19]{Kie19}
Emanuel Kieronski.
\newblock One-dimensional guarded fragments.
\newblock In {\em 44th International Symposium on Mathematical Foundations of
  Computer Science, {MFCS} 2019}, volume 138 of {\em LIPIcs}, pages
  16:1--16:14, 2019.
\newblock \href {https://doi.org/10.4230/LIPIcs.MFCS.2019.16}
  {\path{doi:10.4230/LIPIcs.MFCS.2019.16}}.

\bibitem[Kie23]{Kie23}
Emanuel Kieronski.
\newblock A uniform one-dimensional fragment with alternation of quantifiers.
\newblock In {\em Fourteenth International Symposium on Games, Automata,
  Logics, and Formal Verification, GandALF 2023}, volume 390 of {\em {EPTCS}},
  pages 1--15, 2023.
\newblock \href {https://doi.org/10.4204/EPTCS.390.1}
  {\path{doi:10.4204/EPTCS.390.1}}.

\bibitem[KK14]{KK14}
Emanuel Kieronski and Antti Kuusisto.
\newblock Complexity and expressivity of uniform one-dimensional fragment with
  equality.
\newblock In {\em Mathematical Foundations of Computer Science 2014 - 39th
  International Symposium, {MFCS} 2014, Proceedings, Part {I}}, volume 8634 of
  {\em Lecture Notes in Computer Science}, pages 365--376, 2014.
\newblock \href {https://doi.org/10.1007/978-3-662-44522-8\_31}
  {\path{doi:10.1007/978-3-662-44522-8\_31}}.

\bibitem[KK15]{KK15}
Emanuel Kieronski and Antti Kuusisto.
\newblock Uniform one-dimensional fragments with one equivalence relation.
\newblock In {\em 24th {EACSL} Annual Conference on Computer Science Logic,
  {CSL} 2015}, volume~41 of {\em LIPIcs}, pages 597--615, 2015.
\newblock \href {https://doi.org/10.4230/LIPICS.CSL.2015.597}
  {\path{doi:10.4230/LIPICS.CSL.2015.597}}.

\bibitem[KLPS20]{KLPS20}
Andreas Krebs, Kamal Lodaya, Paritosh~K. Pandya, and Howard Straubing.
\newblock Two-variable logics with some betweenness relations: Expressiveness,
  satisfiability and membership.
\newblock {\em Log. Methods Comput. Sci.}, 16(3), 2020.
\newblock \href {https://doi.org/10.23638/LMCS-16(3:16)2020}
  {\path{doi:10.23638/LMCS-16(3:16)2020}}.

\bibitem[KMPT14]{KMPHT14}
Emanuel Kieronski, Jakub Michaliszyn, Ian Pratt{-}Hartmann, and Lidia Tendera.
\newblock Two-variable first-order logic with equivalence closure.
\newblock {\em {SIAM} J. Comput.}, 43(3):1012--1063, 2014.
\newblock \href {https://doi.org/10.1137/120900095}
  {\path{doi:10.1137/120900095}}.

\bibitem[KMW62]{KMW62}
A.~S. Kahr, Edward~F. Moore, and Hao Wang.
\newblock Entscheidungsproblem reduced to the $\forall \exists \forall$ case.
\newblock {\em Proc. Nat. Acad. Sci. U.S.A.}, 48:365--377, 1962.

\bibitem[KO12]{KO12}
Emanuel Kiero\'{n}ski and Martin Otto.
\newblock Small substructures and decidability issues for first-order logic
  with two variables.
\newblock {\em Journal of Symbolic Logic}, 77:729--765, 2012.
\newblock \href {https://doi.org/10.2178/JSL/1344862160}
  {\path{doi:10.2178/JSL/1344862160}}.

\bibitem[KR21]{KR21}
Emanuel Kieronski and Sebastian Rudolph.
\newblock Finite model theory of the triguarded fragment and related logics.
\newblock In {\em 36th Annual {ACM/IEEE} Symposium on Logic in Computer
  Science, {LICS} 2021}, pages 1--13, 2021.
\newblock \href {https://doi.org/10.1109/LICS52264.2021.9470734}
  {\path{doi:10.1109/LICS52264.2021.9470734}}.

\bibitem[Kuu16]{Kuu16}
Antti Kuusisto.
\newblock On the uniform one-dimensional fragment.
\newblock In {\em 29th International Workshop on Description Logics, Cape Town,
  South Africa, April 22-25, 2016.}, 2016.

\bibitem[Lew80]{Lew80}
Harry~R. Lewis.
\newblock Complexity results for classes of quantificational formulas.
\newblock {\em Journal of Computer and System Sciences}, 21(3):317 -- 353,
  1980.
\newblock \href {https://doi.org/10.1016/0022-0000(80)90027-6}
  {\path{doi:10.1016/0022-0000(80)90027-6}}.

\bibitem[Mas71]{Mas71}
Sergei~J. Maslov.
\newblock The inverse method for establishing deducibility for logical calculi.
\newblock {\em The Calculi of Symbolic Logic I: Proceedings of the Steklov
  Institute of Mathematics}, 98, 1971.

\bibitem[Mor75]{Mor75}
Michael Mortimer.
\newblock On languages with two variables.
\newblock {\em Zeitschrift f{\"{u}}r Mathematische Logik und Grundlagen der
  Mathematik}, 21:135--140, 1975.

\bibitem[MP15]{MP15}
Fabio Mogavero and Giuseppe Perelli.
\newblock Binding forms in first-order logic.
\newblock In {\em 24th {EACSL} Annual Conference on Computer Science Logic,
  {CSL} 2015}, volume~41 of {\em LIPIcs}, pages 648--665, 2015.
\newblock \href {https://doi.org/10.4230/LIPICS.CSL.2015.648}
  {\path{doi:10.4230/LIPICS.CSL.2015.648}}.

\bibitem[Pra10]{P-H10}
Ian Pratt{-}Hartmann.
\newblock The two-variable fragment with counting revisited.
\newblock In {\em Logic, Language, Information and Computation, 17th
  International Workshop, WoLLIC 2010}, pages 42--54, 2010.
\newblock \href {https://doi.org/10.1007/978-3-642-13824-9\_4}
  {\path{doi:10.1007/978-3-642-13824-9\_4}}.

\bibitem[Pra15]{PH15}
Ian Pratt{-}Hartmann.
\newblock The two-variable fragment with counting and equivalence.
\newblock {\em Math. Log. Q.}, 61(6):474--515, 2015.
\newblock \href {https://doi.org/10.1002/malq.201400102}
  {\path{doi:10.1002/malq.201400102}}.

\bibitem[Pra21]{P-H21}
Ian Pratt{-}Hartmann.
\newblock Fluted logic with counting.
\newblock In {\em 48th International Colloquium on Automata, Languages, and
  Programming, {ICALP} 2021}, volume 198 of {\em LIPIcs}, pages 141:1--141:17,
  2021.
\newblock \href {https://doi.org/10.4230/LIPIcs.ICALP.2021.141}
  {\path{doi:10.4230/LIPIcs.ICALP.2021.141}}.

\bibitem[PST97]{PST97}
Leszek Pacholski, Wies\l{}aw Szwast, and Lidia Tendera.
\newblock Complexity of two-variable logic with counting.
\newblock In {\em 12th Annual {IEEE} Symposium on Logic in Computer Science,
  {LICS} 1997}, pages 318--327. IEEE, 1997.
\newblock \href {https://doi.org/10.1109/LICS.1997.614958}
  {\path{doi:10.1109/LICS.1997.614958}}.

\bibitem[PST19]{P-HST19}
Ian Pratt{-}Hartmann, Wieslaw Szwast, and Lidia Tendera.
\newblock The fluted fragment revisited.
\newblock {\em J. Symb. Log.}, 84(3):1020--1048, 2019.
\newblock \href {https://doi.org/10.1017/JSL.2019.33}
  {\path{doi:10.1017/JSL.2019.33}}.

\bibitem[Qui69]{Q69}
Willard~V. Quine.
\newblock On the limits of decision.
\newblock In {\em Proceedings of the 14th International Congress of
  Philosophy}, volume III, pages 57--62, 1969.

\bibitem[R{\v{S}}18]{RS18}
Sebastian Rudolph and Mantas {\v{S}}imkus.
\newblock The triguarded fragment of first-order logic.
\newblock In {\em 22nd International Conference on Logic for Programming,
  Artificial Intelligence and Reasoning, {LPAR-22}}, volume~57 of {\em EPiC
  Series in Computing}, pages 604--619, 2018.
\newblock \href {https://doi.org/10.29007/M8TS} {\path{doi:10.29007/M8TS}}.

\bibitem[Sco62]{Sco62}
Dana Scott.
\newblock A decision method for validity of sentences in two variables.
\newblock {\em Journal Symbolic Logic}, 27:477, 1962.

\bibitem[tCS13]{StC13}
Balder ten Cate and Luc Segoufin.
\newblock Unary negation.
\newblock {\em Logical Methods in Comp. Sc.}, 9(3), 2013.
\newblock \href {https://doi.org/10.2168/LMCS-9(3:25)2013}
  {\path{doi:10.2168/LMCS-9(3:25)2013}}.

\end{thebibliography}

\end{document}